\documentclass[journal, 11pt]{IEEEtran}

\usepackage{float}
\usepackage{graphicx}
\usepackage{epstopdf}%
\usepackage{subfigure}
\usepackage{algorithmicx}
\usepackage{algpseudocode}
\usepackage{algorithm}
\usepackage{amsmath}
\usepackage{amsfonts,amssymb}
\usepackage{amsthm}
\usepackage[font=scriptsize,caption=false,labelsep=space]{subfig}

\newtheorem{lemma}{Lemma}
\newtheorem{proposition}{Proposition}

\graphicspath{{./figures/}}
\usepackage{xcolor}

\newcommand{\stkout}[1]{\ifmmode\text{\sout{\ensuremath{#1}}}\else\sout{#1}\fi}


\begin{document}

\title{A Markov Variation Approach to Smooth Graph Signal Interpolation}
\author{Ayelet Heimowitz and Yonina C. Eldar \thanks{Ayelet Heimowitz is with the Program in Applied and Computational Mathematics, Princeton University, Princeton, NJ.  Yonina  Eldar is with the Faculty of Math and computer science, Weizmann institute of Science, Rehovot, Israel. e-mail: ayeltg@gmail.com, yonina@weizmann.ac.il;}
\thanks{This project has received funding from the European Union's Horizon 2020 research and innovation program under grant agreement No. 646804-ERC- COG-BNYQ and from the Israel science foundation under grant No. 0100101.}
\thanks{This research was carried out while the first author was a postdoctoral  researcher at the Department of Electrical Engineering, Technion, Haifa,  Israel.}}
\date{}
\maketitle

\abstract

In this paper we present the Markov variation, a  smoothness measure which offers a probabilistic 
interpretation of graph signal smoothness.  This measure is then used to develop an optimization framework for 
graph signal interpolation. Our approach is based on diffusion embedding vectors and the connection between 
diffusion maps and signal processing on graphs. As diffusion embedding 
vectors may be expensive to compute for large graphs,  we present a computationally efficient method, based on the Nystr\"{o}m extension, 
for interpolation of signals over a graph.  
We demonstrate our approach on the MNIST dataset and a dataset of daily average temperatures around the US. 
We show that our method outperforms state of the art graph signal interpolation techniques on both datasets, 
and that our computationally efficient reconstruction achieves slightly reduced accuracy with a large computational speedup. 

{\normalfont
\section{Introduction}

With the advent of the world wide web and the move to cloud based computing, 
massive amounts of data have become increasingly available. 
The data may be collected from sources  
such as social networks, government 
agencies, commercial and academic bodies and more. Such data sets can include, 
for example, blogs, 
temperature measurements and information 
on customer preferences. 
Graphs are a popular model for the 
underlying geometry of data.   
Each data element (point) is  represented as a node, and the pairwise connections between the 
different points are modeled as edges. 

As an example, consider a data set of images of written digits, \textit{e.g.} the
MNIST data set \cite{lecun1998mnist}. Each data point is an image of a digit, and
is represented as a node. The similarity between two points 
(\textit{i.e.} two images of digits)
is expressed through the edge weights. In the context of
social networks, each user is a node in the graph, and the relationships between
users are modeled as edge weights \cite{hoff2001social}.
Such relationships may be for example friendship or collaboration.

Graph signals  are 
signals defined over irregular domains represented as weighted graphs \cite{shuman2013review, Sandryhaila2013shift, Sandryhaila2014filter, chen2015sampling, dong2016smooth}.  
The signal is defined as a mapping of each node in a graph to a scalar
\cite{Sandryhaila2013shift, Sandryhaila2014filter, chen2015sampling, sandryhaila2014bigdata}, and can be 
represented as a vector in  $\mathbb{R}^N$.  

In this paper we focus on a subclass of graph signals, namely smooth graph signals. 
Such signals are a mapping of each node  to a scalar (real or complex) such that the geometry of the 
graph is adhered to. 
A vector (signal) that obeys the graph geometry  will be smooth over the edges of the graph. 
This smoothness  is determined through a  measure which 
assigns  a numerical value detailing the change of the signal over the 
graph edges.  
Smoothness criteria have been discussed, for example, in \cite{dong2016smooth, zhou2004smooth}. 
Here we suggest a measure based on the Markov variation, which is a  probabilistic smoothness measure for graph signals. 
The probabilistic nature of our criterion is due to our use of the Markov matrix $\mathbf{P}$ to encode the geometry of the 
graph. The $ij$th entry of this matrix can be considered as the probability to transition from node $i$ to node $j$. 

Our  graph signal smoothness criterion is 
used  to suggest three methods  
for graph signal interpolation. Graph signal interpolation, or semi-supervised learning 
of graph signals, is the problem where a graph signal is known over a subset of nodes 
(the sampled nodes), and 
the goal is to recover the entire  signal from its samples.  The importance of this problem lies in 
the fact that for large graphs computing or measuring the entire signal may be very expensive. 

The first interpolation method we suggest 
uses our smoothness criterion to define a system of linear equations 
over the sampled nodes.  
These equations impose smoothness over each of the samples individually.  
Therefore, all possible solutions must be smooth over the neighborhood 
of the sampled nodes.   
Next, we propose an extension to our suggested method, where the interpolation is {performed} iteratively.  
In iteration $i$, {we interpolate} over all nodes 
in the $0,1,\dots,i$th neighborhoods of any of the 
 sampled nodes. That is, in the first iteration, the interpolation is done over the sampled nodes. In the second 
iteration, we interpolate over the sampled nodes and all nodes that are adjacent to a sampled node, and so on.
In this way, the final solution is guaranteed to be smooth over all the nodes.

Both  interpolation methods discussed above necessitate  computation of  the spectral decomposition of the graph shift operator, 
which  may be infeasible for big data.  We therefore introduce 
a computationally efficient approximation of our method which is derived  
using properties of the Markov matrix and a variation on the Nystr\"{o}m extension which we previously introduced in \cite{heimowitz2017diffusion, heimowitz2018nystrom}. 
This approximation achieves good accuracy in comparatively short runtimes, is feasible for massive datasets and makes no assumption on the 
sampling of the graph signals.  As with the first method we discuss, this approximation can also be extended to be implemented iteratively. 
The  Nystr\"{o}m extension variation we suggest can also be utilized 
for spectral regression~\cite{keller2011regression} and for the method in \cite{Narang2013ssl}  with minor modifications. 
These modifications are necessary since our smoothness interpolation method, as well as spectral regression, 
use the Markov matrix to define the connectivity of the graph while \cite{Narang2013ssl} uses the normalized graph Laplacian.

We provide a mathematical comparison between our interpolation techniques with existing state-of-the-art methods~\cite{chen2015sampling, keller2011regression, segarra2016ssl, Narang2013ssl, gadde2014smooth, jung2019message} in Section~\ref{subsec:comparison}. Furthermore, 
 in Section~\ref{subsec:Experiment}, we use the sampling suggested by Chen \textit{et al.} \cite{chen2015sampling} 
to  compare our {first two} interpolation methods with  \cite{chen2015sampling} and \cite{jung2019message} on the MNIST dataset 
of hand-written digits \cite{lecun1998mnist}. We also use this sampling to compare our iterative method with spectral regression \cite{keller2011regression}, noiseless inpainting~\cite{chen2015var, chen2016reconstruction, ma2015jmlr} on both a synthetic dataset and a dataset of 
temperature measurements across the US \cite{gsod2011dataset}.  
We show that our interpolation techniques outperform all these methods, and {that our iterative interpolation} achieves good  results even when a small number of samples ($10-20$) {is} used. 

This paper is organized as follows. Section \ref{sec:formulation} 
 contains background on the field of signal processing on graphs and the graph signal interpolation problem. In Section \ref{subsec:variation}  we introduce our 
smoothness measure for signals defined on graphs. 
In Section 
\ref{sec:mainIdea} we present our framework for graph signal interpolation. We present our computationally efficient variant in Section~\ref{subsec:bigData}. Finally, 
experimental results demonstrating our proposed graph interpolation methods are presented in Section~\ref{subsec:Experiment}.

\section{Problem Formulation}
\label{sec:formulation}

A graph is denoted as $\mathcal{G} = \langle \mathcal{V}, \mathcal{E}\rangle$
where $\mathcal{V}$ is the set of nodes and $\mathcal{E}$ is the set of
edges. For weighted graphs we denote the affinity matrix containing edge
weights as $\mathbf{W}$. 
The $ij$th element of the affinity matrix specifies the weight of an edge between node $v_i$ 
and node $v_j$. If no edge exists between these nodes, then  $W_{i,j}$ is set to $0$.

A signal over a graph is 
defined in the literature as a mapping from 
each node $v_i$ to a real or complex scalar value $s_i$ \cite{Sandryhaila2013shift, Sandryhaila2014filter, chen2015sampling, sandryhaila2014bigdata}.   
The pairwise (edge) information of the graph is contained in the 
graph shift operator $\mathbf{A}$ \cite{Sandryhaila2013shift}. The graph shift operator is a 
weighted adjacency matrix where the $ij$th entry corresponds to 
the pairwise relationship between nodes $v_i$ and $v_j$.
 This operator may be the affinity matrix $\mathbf{W}$, the graph Laplacian $\mathbf{L}$ or the Markov matrix $\mathbf{P}$.

The graph shift operator is used to generalize operations in signal processing to graphs. One such operation is the graph shift operation \cite{Sandryhaila2014filter}, which is defined as
\begin{equation}\label{equ:graph_shift}
  \tilde{\mathbf{s}} = \mathbf{A} \mathbf{s}. 
\end{equation} 
This operation redistributes the graph signal at each node according to its
neighborhood and is a generalization of time shifts \cite{Sandryhaila2013shift}.   

The graph shift operator is also  used in the definition of the 
graph Fourier transform (GFT) \cite{Sandryhaila2013shift}, which is defined as
\begin{equation}\label{equ:gft}
  \hat{\mathbf{s}} = \mathbf{V}^{-1} \mathbf{s}.
\end{equation}
If the graph shift is diagonalizable then $\mathbf{V}$ is the matrix containing in its columns the eigenvectors of the graph shift operator.
Otherwise, $\mathbf{V}$ is the matrix containing in its columns the generalized eigenvectors of the graph shift operator.
The vector $\hat{\mathbf{s}}$ is the spectrum of the graph signal. When this spectrum 
contains $k$ nonzero entries, we say that the  graph signal is $k$-bandlimited \cite{chen2015sampling}.  

In this paper, we focus on smooth graph signals. 
Such signals are mappings from 
each node $v_i$ to a real or complex scalar value $s_i$, such that 
the vector $\mathbf{s} = \begin{bmatrix}  s_1 & \cdots & s_{N} \end{bmatrix}^T\in \mathbb{C}^N$ 
is smooth over the graph. Under this definition, the  geometry of the graph will contain information about the graph signal, and the graph signal 
 contains geometric information. 
 
 We consider the problem of graph signal interpolation, where a smooth graph signal is recovered from its samples and the known graph structure. We denote the set of $r$ sampled nodes as  $\mathcal{M}$  and the vector of samples as $\mathbf{s}_{\mathcal{M}} \in \mathbb{R}^{r \times 1}$.  For this problem,  
perfect reconstruction is possible for $k$ bandlimited graph signals under conditions formulated in~\cite{chen2015sampling, Narang2013ssl}. Specifically, 
a bandlimited  graph signal can be perfectly recovered from its samples if the matrix  
produced by sampling  the  $k$ eigenvectors  at the rows corresponding to the  
known graph signal is  invertible.

In the following sections, we will present three  algorithms for smooth graph signal interpolation. In two of these solutions, perfect reconstruction is guaranteed under the conditions specified above. Our third method is characterized by reduced computational complexity and, as a result, fast runtimes. Our rechniques are based on the Markov variation, a  smoothness measure which we motivate and introduce in the next section.

\section{Smooth Graph Signals}
\label{subsec:variation}

As mentioned in Section~\ref{sec:formulation}, we consider the problem of graph signal interpolation, where a  smooth graph signal is recovered from its samples. 
The smoothness of a signal  is to be determined through a  measure which 
assigns  a numerical value detailing the change of the signal over the 
graph edges.  Examples of such measures include the total variation measure, defined as 
in \cite{Sandryhaila2014filter},
\begin{equation} \label{equ:tv}
TV \left( \mathbf{s} \right) = \Vert \mathbf{s} - \tilde{\mathbf{A}}\mathbf{s} \Vert_p,
\end{equation}
where $\tilde{\mathbf{A}}$ is a normalization of the graph shift operator (the weighted adjacency matrix) such that the largest 
magnitude eigenvalue is equal to one. 

An alternative smoothness measure, related to edge derivatives \cite{dong2016smooth, zhou2004smooth}, is given by
\begin{equation} \label{equ:measure2}
\mathbf{s}^T \mathbf{\mathcal{L}} \mathbf{s} = \frac{1}{2} \sum_{i=1}^{N} \sum_{m=1}^{N} W_{i, m} \left(s_i - s_m \right)^2
\end{equation}
where $\mathbf{W}$ is the symmetric affinity matrix, $\mathbf{\mathcal{L}}$ is the unnormalized graph Laplacian $\mathbf{\mathcal{L}} = \mathbf{D} - \mathbf{W}$ and $\mathbf{D}$ so the diagonal matrix that contains in its diagonal the degree of each node.

There are two properties of smoothness we would like to ensure in our smoothness measure. First, we would like the smoothness measure to reach a global minimum for a constant graph signal (that is, a graph signal that maps every node to the same scalar value). The second property is that the smoothness measure can distinguish between graph signals that are smooth across each edge of the graph individually, and graph signals that are smooth 
across all  incident edges. 

To test the first property, we consider all signals of the form
\begin{equation*}
\mathbf{s}_{1} = c \mathbf{1},
\end{equation*}
where $\mathbf{1}$ is the all ones vector and $c \in \mathbb{C}$. Signals of this form do not change between 
any two nodes in the graph, and are therefore the smoothest possible graph signals. 
The measure~\eqref{equ:measure2} will indeed reach a global minimum for such signals as $\mathbf{\mathcal{L}}$ is known to be positive semi-definite and as such
$$\mathbf{s}^T \mathbf{\mathcal{L}} \mathbf{s}  \ge 0$$
and, additionally
\begin{equation}
\mathbf{s}_{1}^T \mathbf{\mathcal{L}} \mathbf{s}_{1} = \frac{1}{2} \sum_{i=1}^{N} \sum_{m=1}^{N} W_{i, m} \left(c - c \right)^2 = 0.
\end{equation}
However, the total variation smoothness measure may violate this criterion. The total variation of $\mathbf{s}_1$ is
\begin{equation} 
TV \left( \mathbf{s}_{1} \right) =  \Vert c \mathbf{1} -  c \tilde{\mathbf{D}} \mathbf{1}  \Vert_p,
\end{equation}
where $\tilde{\mathbf{D}}$ is the matrix containing in its diagonal the degrees of $\tilde{\mathbf{A}}$. This will  equal $0$ when 
$\tilde{\mathbf{D}}$ is the identity matrix. Otherwise, there is no guarantee that the total variation  
 will reach a global minimum for signals of the form  $\mathbf{s}_1$. 
 
 As for the second property, we consider the following affinity matrix
\begin{equation*}
\mathbf{A} = \begin{bmatrix}
0 & 1 & 1 & 1 & 1\\
1 & 0      & 0      & 0      & 0    \\ 
1 & 0      & 0      & 0      & 0    \\ 
1 & 0      & 0      & 0      & 0    \\ 
1 & 0      & 0      & 0      & 0    \\ 
\end{bmatrix},
\end{equation*}
{and the  graph signals}, 
\begin{gather*}
\mathbf{s}_2 = \begin{bmatrix} 0 & -2 & -2 & 2 & 2 \end{bmatrix}^T,\\
\mathbf{s}_3 = \begin{bmatrix} 0 & 2 & 2 & 2 & 2 \end{bmatrix}^T.
\end{gather*}
When evaluating the change of signal over each edge independently, these signals are equally smooth. 
However, when considering the change of signal across all edges incident to each node, 
$\mathbf{s}_2$ is smoother than $\mathbf{s}_3$. According to~\eqref{equ:measure2} these signals are  equally smooth, since
\begin{equation}
\mathbf{s}_2^T \mathbf{\mathcal{L}} \mathbf{s}_2 = \frac{1}{2} \sum_{i=1}^{N} \sum_{m=1}^{N} 4 W_{i, m} = \mathbf{s}_3^T \mathbf{\mathcal{L}} \mathbf{s}_3.
\end{equation}
The reason for this is that~\eqref{equ:measure2} takes into account only the magnitude of the difference between the graph signal on adjacent nodes. 
In other words, the change of a graph signal across edge $e_i$ is  independent of the change across all edges incident to $e_i$.  

In this paper we suggest an additional smoothness measure, the Markov variation. This measure has a probabilistic nature, which ensures that both 
aforementioned qualities are attained. That is, the Markov variation  will reach a global minimum for signals of the form  $\mathbf{s}_1$, and, additionally, 
takes into account the changes across all incident edges. 
The Markov variation is the measure
\begin{equation}
\Vert \mathbf{s} - \mathbf{D}^{-1} \mathbf{A} \mathbf{s} \Vert,
\end{equation}
where $\mathbf{D}$ is the diagonal matrix containing the degrees of $\mathbf{A}$. This measure is similar to total variation, the difference being in the suggested normalization of the graph shift operator. For $\mathbf{A}=\mathbf{P}$ the Markov variation and the total variation will equal.

To gain intuition into this measure, we 
consider a smooth graph signal $\mathbf{s} \in \mathbb{R}^N$.  
Since $\mathbf{s}$ should map closely connected nodes to  
similar values, we can think of the graph signal at the nodes neighboring  $v_i$ (\textit{i.e.} 
 $\{s_j \}_{j \in \mathcal{N}_i}$)  as defining a distribution over $s_i$. We therefore model 
 $s_i$ as
\begin{equation} \label{equ:consistModel2}
s_i = \sum_{m \in \mathcal{N}_i}  P_{i, m} s_m + \epsilon \left( s_i \right),
\end{equation}
where  $\mathbf{P}$ is the Markov transition matrix, $$\mathbf{P} = \mathbf{D}^{-1} \mathbf{A},$$ $ \mathcal{N}_i$ is the set of nodes adjacent to $v_i$ and $\epsilon$ is 
the error.

The model  in (\ref{equ:consistModel2}) consists of two terms. 
The first  
is an estimate of the graph signal $s_i$ 
based only on the neighboring nodes and transition probabilities. This is a Markovian model, where the assumption is that when the graph 
signals at neighboring nodes are known, there is no dependence on non-neighboring nodes. Since we focus on graph signals whose mappings  
conform to the geometry of the graph, 
the transition probabilities between two neighboring nodes can be thought of as an approximation of the probability of both nodes having the same graph signal. 
Thus, for any smooth graph signal $\mathbf{s}$ we expect $s_i - \sum_{m \in \mathcal{N}_i}  P_{i, m} s_m$ to be small. 
This difference corresponds to the second term in (\ref{equ:consistModel2}) which is an error term  $\epsilon$ 
that explains variations from the 
weighted sum of neighbors. 

Our suggested measure, which we call the Markov variation, is the  norm of the error term
\begin{equation} \label{generalMeasure}
MV \left( \mathbf{s} \right) = \Vert \epsilon \left( \mathbf{s} \right) \Vert_p = \Vert \mathbf{s} - \mathbf{D}^{-1} \mathbf{A} \mathbf{s} \Vert = \Vert \mathbf{s} - \mathbf{P} \mathbf{s} \Vert_p.
\end{equation}
For example, using the $\ell_1$ norm we obtain
\begin{equation} \label{equ:ourMeasure}
\Vert \mathbf{s} - \mathbf{P} \mathbf{s} \Vert_1 =  \sum_{i=1}^{N}  \frac{1}{d\left( v_i \right)} \left| \sum_{m=1}^{N}  W_{i, m} \left(s_i - s_m \right) \right|.
\end{equation}
We  define  a smooth graph signal as a vector $\mathbf{s} \in \mathbb{R}^N$ with low Markov variation, \textit{i.e.},
\begin{equation} \label{equ:definition}
\Vert \mathbf{s} - \mathbf{P} \mathbf{s} \Vert_p < \eta,
\end{equation}
where $\eta$ is determined according to the number of nodes in the graph. 

The Markov variation bears some similarity to both \eqref{equ:tv} and \eqref{equ:measure2}, 
while offering a probabilistic interpretation.  Due to its probabilistic nature, this measure reaches a global minimum for $\mathbf{s}_1$, that is
\begin{equation} \label{equ:mv}
MV \left( \mathbf{s}_{1} \right) = \Vert c \mathbf{1} - c \mathbf{P 1} \Vert = \Vert c \mathbf{1} - c \mathbf{1} \Vert = 0.
\end{equation}
Furthermore, when considering the change of signal across all edges incident to each node, 
$\mathbf{s}_2$ is smoother than $\mathbf{s}_3$
\begin{equation} 
MV \left( \mathbf{s}_{2} \right) =  \Vert \mathbf{s}_2 - \mathbf{P}\mathbf{s}_2 \Vert_2 = 4,
\end{equation}
\begin{equation} 
MV \left( \mathbf{s}_{3} \right) =  \Vert \mathbf{s}_3 - \mathbf{P}\mathbf{s}_3 \Vert_2 = 4.47.
\end{equation}
Once again, the Markov variation attains this property due to its probabilistic nature.

\section{Graph Signal Interpolation}
\label{sec:mainIdea}

We now show how  the Markov variation can be used for interpolation of smooth graph signals from $r$ samples,  
where each sample is a mapping of a node 
to a known scalar. To this end, we first provide a spectral interpretation of the Markov variation, and connect it to diffusion maps~\cite{coifman2006diffusion}. We 
then use these properties to show that the spectrum of a smooth graph signal is naturally sparse. Based on these insights, in Section~\ref{subsec:method1}, we  formulate the interpolation of  smooth graph signals as a constrained optimization problem.

We denote the set of $r$ sampled nodes as  $\mathcal{M}$  
and the vector of samples   
 as $\mathbf{s}_{\mathcal{M}} \in \mathbb{R}^{r \times 1}$. 
Our goal is to recover a smooth graph signal $\mathbf{s}$ from $\mathbf{s}_{\mathcal{M}}$ 
using the known graph structure.

\subsection{Spectral Interpretation of the Markov Variation}

The  Markov variation expresses an equivalence between smoothness measured 
over the edges of the graph and smoothness measured over diffusion embedding vectors~\cite{coifman2006diffusion, heimowitz2017diffusion}, which are defined as
\begin{equation}
{\Psi}_t \left( i \right) = \begin{bmatrix}
\lambda_1^t \psi_1 \left( i \right) \\
\lambda_2^t \psi_2 \left(i \right) \\
\vdots \\
\lambda_N^t \psi_N \left(i \right) \end{bmatrix}, \quad i=1,\dots,N,
\label{equ:embedding}
\end{equation}
where $\lambda_i$ and $\psi_i$ are the $i$th eigenvalue and eigennvector of the Markov matrix $\mathbf{P}$, correspondingly, and $t$ is a scale factor.

Since  the Markov matrix is diagonalizable for undirected graphs (see  Appendix A),~\eqref{generalMeasure} can be written as 
\begin{equation}
\Vert \mathbf{s} - \mathbf{P}\mathbf{s} \Vert_p = \Vert \mathbf{s} - \mathbf{V} \mathbf{\Lambda} \mathbf{V}^{-1} \mathbf{s} \Vert_p
\end{equation}
where $\mathbf{\Lambda}$ is a diagonal matrix that contains the eigenvalues of the Markov matrix.  
In general, the diffusion embedding vectors can be expressed as 
$\mathbf{\Psi}_t = \mathbf{V} \mathbf{\Lambda}^t$, where the $i$th row  of $\mathbf{\Psi}_t $ 
equals
\begin{equation}
\Psi_t^T \left( i \right) =  \begin{bmatrix}  \lambda_1^t \psi_1 \left( i \right) & \lambda_2^t \psi_2 \left( i \right) & \cdots & \lambda_N^t \psi_N \left( i \right) \ \end{bmatrix}.
\end{equation}
The $\ell_p$ norm of  (\ref{equ:definition}) can thus be written as  
\begin{equation} \label{equ:embedSmooth3}
\Vert \mathbf{s} - \mathbf{P}\mathbf{s} \Vert_p^p = 
\sum_{i=1}^N \vert  s_i -   \Psi_1^T \left( i \right) \hat{\mathbf{s}}
\vert^p < \eta^p,
\end{equation}
which implies that for a vector $\mathbf{s}$ to be a smooth graph signal, $s_i$ must be close to $s_j$ 
if  $\Psi_1 \left(i \right)$ and $\Psi_1 \left( j \right)$ are close (in the $\ell_2$ sense).

We conclude that on the one hand the Markov variation can be expressed as a connection between 
the graph signal and the geometry of the graph in the graph domain. On the other hand, 
the smoothness function can be expressed as a connection between the spectrum of the 
graph signal and the diffusion embedding vectors in the frequency domain.

Another important conclusion can be obtained from
\begin{equation} \label{equ:embedSmooth2}
\Vert \mathbf{s} - \mathbf{P}\mathbf{s} \Vert_p^p = 
 \Vert \mathbf{V} \left( \mathbf{I}_N - \mathbf{\Lambda}  \right) \hat{\mathbf{s}}  \Vert_p^p < \eta^p.
\end{equation}
Since the largest eigenvalue of the Markov matrix is $1$ (see Appendix A) and the magnitude 
of the smaller eigenvalues  is often  $0$, 
the entries of $\hat{\mathbf{s}}$ that are related to the highest eigenvalues do not contribute much  
to the sum (\ref{equ:embedSmooth2}). The entries of $\hat{\mathbf{s}}$ that  correspond to 
the lower eigenvalues  have a higher impact on the sum (\ref{equ:embedSmooth2}). 
This means that, in order for a signal $\mathbf{s}$ to be a smooth graph signal, many of the entries of $\hat{\mathbf{s}}$ that correspond to 
the lower valued eigenvalues must be negligible. In other words, the spectrum of a smooth graph 
signal is naturally approximately sparse.\footnote{We note that 
while we show this only for the case where the graph shift is the Markov matrix, it is true also for general graph shifts \cite{chen2015sampling}.}

\subsection{Interpolation by Smoothness}
\label{subsec:method1}

As the graph signal is smooth, we conclude from (\ref{equ:embedSmooth3}) that  
for each node $i$, 
\begin{equation} \label{equ:smooth2}
s_i \approx  \Psi_1^T \left( i \right) \hat{\mathbf{s}}.
\end{equation}
If the signal $\mathbf{s}$ is the smoothest possible signal according to the Markov variation (\textit{i.e.}, the error term is $0$), then
\begin{equation} \label{equ:smooth3}
s_i =  \Psi_1^T \left( i \right) \hat{\mathbf{s}},
\end{equation}
which leads to the following system of equations 
\begin{equation} \label{equ:system}
\begin{bmatrix} s_1 \\ s_2 \\ \ \vdots \\ s_N \end{bmatrix}  =
\begin{bmatrix} 
\lambda_1 \psi_1 \left( 1 \right) & \lambda_2 \psi_2 \left( 1 \right) & \cdots & \lambda_N \psi_N \left( 1 \right)\\
\lambda_1 \psi_1 \left( 2 \right) & \lambda_2 \psi_2 \left( 2 \right) & \cdots & \lambda_N \psi_N \left( 2 \right)\\
\vdots & \vdots & \ddots & \vdots\\
\lambda_1 \psi_1 \left( N \right) & \lambda_2 \psi_2 \left( N \right) & \cdots & \lambda_N \psi_N \left( N \right)
\end{bmatrix} 
\begin{bmatrix} \hat{s}_1 \\ \hat{s}_2 \\ \vdots \\ \hat{s}_N \end{bmatrix}.
\end{equation}
Out of these $N$ equations, we examine those  that correspond to the  known graph signal
\begin{equation} \label{equ:underdetermine}
\mathbf{s}_{\mathcal{M}} = 
\begin{bmatrix} 
\lambda_1 \psi_1 \left( \mathcal{M} \right) & \lambda_2 \psi_2 \left( \mathcal{M}  \right) & \cdots & \lambda_N \psi_N \left( \mathcal{M}  \right)\\
\end{bmatrix} 
\hat{\mathbf{s}},
\end{equation}
where $\psi_i \left( \mathcal{M} \right)$ is the sub-vector of $\psi_i$ that  
contains only the entries at the set of indices $\mathcal{M}$.

The solution of (\ref{equ:underdetermine}) is not unique. 
One such solution is for example the least squares vector, 
\begin{equation}
\arg \underset{\mathbf{\hat{s}}}{\min} \Vert \mathbf{A} \hat{\mathbf{s}} - \mathbf{s}_{\mathcal{M}} \Vert_2
\end{equation}
where 
\begin{equation}
\mathbf{A} = \begin{bmatrix} 
\lambda_1 \psi_1 \left( \mathcal{M} \right) & \lambda_2 \psi_2 \left( \mathcal{M}  \right) & \cdots & \lambda_N \psi_N \left( \mathcal{M}  \right)
\end{bmatrix}.
\end{equation}

The least squares solution ignores our prior knowledge concerning the sparsity of the graph signal. 
Therefore, rather than using this solution, 
we  search for the 
subset of eigenvectors of the Markov matrix that best explain the known portion of the graph signal. 
This leads to   the following sparse optimization problem:
\begin{multline} \label{equ:sparseOpt_org}
\hat{\mathbf{s}} = \arg \underset{\mathbf{y}}{\min} \Vert \mathbf{y} \Vert_0 \quad \text{such that} \\
\begin{bmatrix} 
 \lambda_1 \psi_1 \left( \mathcal{M} \right) & \lambda_2 \psi_2 \left( \mathcal{M}  \right) & \cdots & \lambda_N \psi_N \left( \mathcal{M}  \right)\\
\end{bmatrix} 
\mathbf{y} = \mathbf{s}_{\mathcal{M}} .
\end{multline}
The solution to \eqref{equ:sparseOpt_org} is the sparse spectrum of a signal that is consistent with 
$\mathbf{S}_\mathcal{M}$ and is smooth in the neighborhood of the nodes in $\mathcal{M}$ (See Appendix B)\footnote{We note that this solution uses the prior knowledge that the spectrum of a smooth graph signal is sparse. It does not explicitly use any knowledge of the likely location of the zeros. We discuss this in the end of this section.}. 
The vector of graph signals is obtained by inserting the solution of \eqref{equ:sparseOpt_org} into \eqref{equ:system}.  

The optimization problem of~\eqref{equ:sparseOpt_org} includes $l_0$ regression which is known to be NP-hard. 
We therefore approximate the solution using $l_1$ regression. We also add to each 
constraint a small tolerance 
 for error in  accordance 
with \eqref{equ:smooth2}. The optimization problem we solve is therefore
\begin{multline}  \label{equ:sparseOpt}
\hat{\mathbf{s}} = \arg \underset{\mathbf{y}}{\min} \Vert \mathbf{y} \Vert_1 \quad \text{such that} \\
\vert \begin{bmatrix} 
 \lambda_1 \psi_1 \left( \mathcal{M} \right) & \lambda_2 \psi_2 \left( \mathcal{M}  \right) & \cdots & \lambda_N \psi_N \left( \mathcal{M}  \right)\\
\end{bmatrix} 
\mathbf{y} - \mathbf{s}_{\mathcal{M}} \vert < \eta.
\end{multline}
{The solution of~\eqref{equ:sparseOpt} is not guaranteed to be unique. If there exist several solutions, one can be chosen arbitrarily.}

We note that, the solution of~\eqref{equ:sparseOpt_org} will be a bandlimited graph signal that is guaranteed to be smooth in the one-hop neighborhood of the sampled nodes. 
In general, a bandlimited graph signal need not be smooth. Rather,  in order for a signal $\mathbf{s}$ to be a smooth graph signal, many of the entries of $\mathbf{\hat{s}}$ that correspond to the lower valued eigenvalues must be negligible. In our proposed solution, we search for the spectrum $\mathbf{\hat{s}}$ with minimal $\ell_1$ norm such that
\begin{equation}
\vert  \hat{s}_1 \lambda_1 \psi_1 \left( \mathcal{M} \right) + \cdots + \hat{s}_N \lambda_N  \psi_N \left( \mathcal{M} \right) - \mathbf{s}_{\mathcal{M}} \vert < \eta.
\end{equation}
Since $\lambda_N \le \cdots \le \lambda_1$, in order for some eigenvector $\psi_i$, which is associated with a low-valued $\lambda_i$, to be consequential in the sum, the value of $\hat{s}_i$ must be large. As we search for $\mathbf{\hat{s}}$ with minimal $\ell_1$ norm, this is an unlikely situation. In other words, the solution of our interpolation method is guaranteed to be smooth in the one-hop neighborhood of the sampled nodes and, in addition, contains a bias towards smooth graph signals.

\subsection{Comparison to Existing Interpolation Methods}
\label{subsec:comparison}

{The works} \cite{keller2011regression, Narang2013ssl, segarra2016ssl, chen2015sampling}  
have  taken a similar approach {to the graph signal} interpolation problem, in that they all formulate the interpolation as a solution to some linear system of equations. However, the system of equations defined here is unique since  
its definition is based on the Markov variation.  
Furthermore, we include a bias towards smooth graph signals, and do not predetermine the sparsity of the spectrum.

In contrast, \cite{keller2011regression, Narang2013ssl, segarra2016ssl, chen2015sampling}  all define a 
set of equations based on the graph Fourier transform (GFT). The solutions to such a system are all 
 graph signals that are consistent with the samples of the graph signal. In order to ensure the interpolation 
returns a smooth graph signal, these methods predetermine the sparsity of the {signal} spectrum. In other words, they 
search for a graph signal that is consistent with the samples and resides in the span of $K$ predetermined eigenvectors 
of the graph shift operator (that is, the $K$ leading eigenvectors, where the value of $K$ is often assumed to be known, or determined according to the magnitude of the eigenvalues).

Spectral regression  \cite{keller2011regression} defines the following system of equations
\begin{multline} \label{equ:sparseOpt2}
\hat{\mathbf{s}} = \arg \underset{\mathbf{y}}{\min} \Vert \mathbf{y} \Vert_1 \quad \text{such that} \\
\begin{bmatrix} 
  \psi_1 \left( \mathcal{M} \right) &  \psi_2 \left( \mathcal{M}  \right) & \cdots &  \psi_K \left( \mathcal{M}  \right)\\
\end{bmatrix} 
\mathbf{y} = \mathbf{s}_{\mathcal{M}} ,
\end{multline}
where $\psi_1, \dots \psi_K$ denote the $K$ eigenvectors of the Markov matrix corresponding to the largest magnitude eigenvalues.  
Narang \textit{et al.} \cite{Narang2013ssl} suggest a method for interpolating bandlimited graph signals 
using the eigenvectors of the normalized graph Laplacian. The interpolation is {performed} on $\mathbf{D}^{\frac{1}{2}} \mathbf{s}$, 
and, similar to spectral regression, is based on a system of equations extracted from the GFT. Mathematically, the system of 
linear equations is
\begin{equation} 
\begin{bmatrix} 
  \psi^L_1 \left( \mathcal{M} \right) &  \psi^L_2 \left( \mathcal{M}  \right) & \cdots &  \psi^L_K \left( \mathcal{M}  \right)\\
\end{bmatrix} 
\mathbf{y} = \mathbf{D}^{\frac{1}{2}} \mathbf{s}_{\mathcal{M}},
\label{equ:narang}
\end{equation}
where $\psi^L_1,\dots, \psi^L$ denote $K$ eigenvectors of the normalized graph Laplacian.  
The solution to the system is computed through linear least squares. 
We note that, once again, the bandlimit of the solution to~\eqref{equ:narang} must be determined before solving the system of equations.

In (\cite{chen2015sampling}, Section 5) 
Chen \textit{et al.} suggest interpolation methods for two clustering applications. Once again, their system of equations is created from the GFT. 
In contrast to the previous systems, here each node is mapped to a vector of length $L$ (the number of clusters) rather than a scalar value. 
This vector is actually an indicator function for its node, meaning that for node $i$  in the first class, the signal will be $\begin{bmatrix} 1 & 0 & \dots & 0 \end{bmatrix}^T$. 
As each node is now mapped to a vector, the  graph signal  is a matrix $\mathbf{S} \in \mathbb{R}^{N \times L}$.  The interpolation is defined as the following optimization problem
\begin{multline} 
\hat{\mathbf{S}} = \arg \underset{\mathbf{Y} \in \mathbb{R}^{K \times L}}{\min} \\ \Vert \text{sign} \left( 
\begin{bmatrix}
 \psi^A_1 \left( \mathcal{M} \right) &  \psi^A_2 \left( \mathcal{M}  \right) & \cdots &  \psi^A_K \left( \mathcal{M} \right)
\end{bmatrix} \mathbf{Y}
\right) - \mathbf{S}{_\mathcal{M}} \Vert_2^2,
\label{equ:chen}
\end{multline}
where $\psi^A_1,\dots psi^A_K$ denote $K$ eigenvectors of the  graph shift and $\mathbf{S}{_\mathcal{M}} \in \mathbb{R}^{r \times L}$ is the matrix of the known portion of the graph signal. 
The optimization problem~\eqref{equ:chen} is solved by logistic regression. 
Here again, the sparsity of the spectrum must be predetermined.

Another method for sampling and reconstruction  of a   known graph signal is presented by 
Sergarra \textit{et al.} \cite{segarra2016ssl}. Here there is an added assumption on the formation model of smooth graph signals. 
Specifically, they assume a graph signal $\mathbf{s}$ is created from a known sparse signal $\mathbf{x}$ as
\begin{equation}
\mathbf{s} = \mathbf{H}\mathbf{x},
\end{equation}
where $\mathbf{H}$ is some graph filter. In this method the assumption is that $\mathbf{s}$ is known and the 
goal is to identify $\mathbf{H}$ and $\mathbf{x}$. This interpolation is {performed} through a system of linear equations 
based on the graph Fourier transform,
\begin{equation}
\hat{\mathbf{s}} = \mathbf{V}_L^{-1} \mathbf{Hx},
\end{equation}
where $\mathbf{V}_L$ denotes the matrix of eigenvectors of the normalized graph Laplacian. For a $K$-bandlimited graph signal, 
this set of $N$ equations can be divided 
into two systems. The first system consists of the $N-K$ equations for which $\hat{s}_i = 0$. These equations can be used to identify the 
coefficients of the graph filter $\mathbf{H}$. The rest of the equations are used to interpolate $\mathbf{x}$. Once again, $K$ must be 
predetermined.

In conclusion, the idea of graph signal interpolation via a system of linear equations is quite popular. However, all the methods we 
discuss above use the graph Fourier transform to define this system of equations. These systems are solved  
over the set of graph signals that comply with a predetermined sparsity of the spectrum. 
Contrary to this, we  derived a system of equations that is based on the Markov variation, which is a smoothness measure. 
Any solution to this system is guaranteed to be {a} smooth graph signal. Of all possible 
solutions, we select the signal with the smallest bandwidth. 
In this way, the sparsity of the graph signal's spectrum need not be predetermined. 
Instead, we determine the bandlimit in a data-driven manner. 

{Another difference between this work and \cite{keller2011regression, Narang2013ssl, segarra2016ssl, chen2015sampling}, 
is that our method~\eqref{equ:sparseOpt} naturally extends to iterative interpolation wherein each iteration is solved by the same vector 
or by a smoother vector than the previous iteration. This is not the case for any of \cite{keller2011regression, Narang2013ssl, segarra2016ssl, chen2015sampling}} as these methods are not based on a smoothness measure.

In Section \ref{subsec:Experiment} we show that our suggested interpolation technique{s} outperform 
state-of-the-art graph signal interpolation methods 
on synthetics data as well as the MNIST data set of hand-written digits  \cite{lecun1998mnist} and a data set of temperature measurements \cite{gsod2011dataset}.

All the above methods require knowledge of the eigendecomposition of the graph shift operator. 
This is a costly operation, and infeasible {for} large graphs. In Section~\ref{subsec:bigData} 
we introduce a method for efficiently estimating the eigenvectors and eigenvalues of the Markov matrix. 
This approach can also be used in  spectral regression, and can be easily adjusted to any method that 
uses a positive semi-definite graph shift.

\subsection{Iterative Interpolation}
\label{subsec:it}
Since the system of equations~\eqref{equ:underdetermine} was created on the basis of a smoothness measure, 
any solution must be smooth in the neighborhood of the sampled nodes.  As the spectrum of a smooth graph signal 
is naturally approximately sparse, of all possible  solutions to~\eqref{equ:underdetermine}, we select {a} solution with minimal $l_1$ norm.  

All smooth signals possess a sparse spectrum. However, not every signal with a sparse spectrum is smooth \cite{chen2015sampling}. 
Therefore, it is conceivable that the solution of~\eqref{equ:sparseOpt} may not be smooth over (one-hop) neighborhoods that do not contain  sampled nodes. 
To prevent this, we can iteratively solve~\eqref{equ:sparseOpt} while introducing in each iteration new nodes into the set $\mathcal{M}$. 
In the last iteration we ensure that every node in the graph is a neighbor of some node in $\mathcal{M}$. In this way we  guarantee that 
any solution in the last iteration is smooth.

Our iterative approach is initialized  with the set of sampled nodes $\mathcal{M}_0$. 
 In the first iteration we recover a signal with sparse spectrum 
that is guaranteed to be smooth (according to the Markov variation) in the (one-hop) neighborhood of $\mathcal{M}_0$. Then, 
in each  iteration $i$, we consider the 
set of sampled nodes  to be $\mathcal{M}_i = \mathcal{N} \left( \mathcal{M}_{i-1} \right)$, 
where $\mathcal{N} \left( \mathcal{M}_{i-1} \right)$ denotes the neighborhood of $\mathcal{M}_{i-1}$ 
(note that, by construction, $\mathcal{M}_{i-1} \subset \mathcal{N} \left( \mathcal{M}_{i-1} \right)$).  
The result of the $i$th iteration is a signal with sparse spectrum that is guaranteed to be 
 smooth (according to the markov variation) in the $i$th neighborhood of each node in $\mathcal{M}_0$. 
The stopping condition is that $\mathcal{M}_{i}$ will equal the set of all nodes $\mathcal{V}$. 
Therefore, the  interpolated graph signal is guaranteed to be smooth over all edges of the graph. 
Our {approach} is summarized in Algorithm \ref{alg:it}.

\begin{algorithm}
\caption{Iterative Interpolation}
\label{alg:it}
\begin{algorithmic}
\State {Let $\mathcal{G} = \langle \mathcal{V}, \mathcal{E} \rangle$ be a graph with sampling set $\mathcal{M}_0$}
\State $i \gets 1$
\Repeat
    \State $\hat{\mathbf{s}}_i \gets$ solve~\eqref{equ:sparseOpt}
    \State $\mathcal{M}_{i+1} \gets  \mathcal{N} \{ \mathcal{M}_i \}$ \% update sampling set for next iteration
    \State $i \gets i+1$
\Until {\textcolor{blue}{}{$\mathcal{M}_{i-1} = \mathcal{M}_{i} $}}
\end{algorithmic}
\end{algorithm}

{In each iteration of Algorithm \ref{alg:it} the solution is either unchanged or  improved. 
This is due to the effect of each update. {Specifically, since} $\mathcal{M}_{i} \subset \mathcal{M}_{i-1}$,  
the set of constraints 
\begin{equation}
\vert
\begin{bmatrix}  \lambda_1 \psi_1 \left( \mathcal{M}_i \right) & \lambda_2 \psi_2 \left( \mathcal{M}_i  \right) & \cdots & \lambda_N \psi_N \left( \mathcal{M}_i  \right)\\
\end{bmatrix} 
\mathbf{y} - \mathbf{s}_{\mathcal{M}_i} \vert < \eta
\end{equation}
is increased {in} each iteration. If the solution $\hat{\mathbf{s}}_i$ is smooth over the neighborhood of $\mathcal{M}_{i+1}$ then 
$\hat{\mathbf{s}}_{i+1} = \hat{\mathbf{s}}_i$. If $\hat{\mathbf{s}}_i$ is not smooth over the neighborhood of $\mathcal{M}_{i+1}$ then 
it is not a solution to
 \begin{multline} \label{equ:iter_i1}
\vert \begin{bmatrix} 
 \lambda_1 \psi_1 \left( \mathcal{M}_{i+1} \right) & \lambda_2 \psi_2 \left( \mathcal{M}_{i+1}  \right) & \cdots & \lambda_N \psi_N \left( \mathcal{M}_{i+1}  \right)\\
\end{bmatrix} \\
\mathbf{y} - \mathbf{s}_{\mathcal{M}_{i+1}} \vert < \eta.
\end{multline}
Instead, the solution of iteration $i+1$ is the vector with smallest $l_1$ norm that solves~\eqref{equ:iter_i1}. This vector is 
smoother than $\hat{\mathbf{s}}_i$ and has equal or higher $l_1$ norm.
}

In Fig.~\ref{fig:toy_example} we present a toy example to illustrate the difference between the one-shot algorithm~\eqref{equ:sparseOpt} and Algorithm~\ref{alg:it}. We randomly select $100$ points in $[0,1]\times [0,1]$, denoted as $x_1,\dots,x_{100}$. We define the affinity matrix $\mathbf{W}$ as
\begin{equation*}
W_{i,j} = 
\begin{cases}
 \mathrm{e}^{-d (x_i, x_j )}, & i \neq j,\\
 0, & i=j,
 \end{cases}
\end{equation*}
where $d(x_i,x_j)$ is the Euclidean distance between $x_i$ and $x_j$. We keep the highest $9$ entries in each row of the affinity matrix and set all other entries to zero. We then symmetrize the affinity matrix as
\begin{equation*}
\mathbf{W}_{\mathrm{s}} = \max (\mathbf{W}, \mathbf{W}^T ).
\end{equation*}
We use this symmetric matrix to compute the Markov matrix. 

It is clear from Fig.~\ref{fig:toy_example} that the result of both our methods is a smooth graph signal. However, the output of~\eqref{equ:sparseOpt} is a graph signal that is determined by the one-hop neighborhood of each node  and a bias towards smoothness. The result of Algorithm~\ref{alg:it} is a smooth graph signal determined by larger neighborhoods and a guarantee of smoothness. For this reason,  our iterative solution has a higher dependence on the closest of the sampled nodes. 

\begin{figure}
\subfigure[]{\includegraphics[width=0.45 \linewidth]{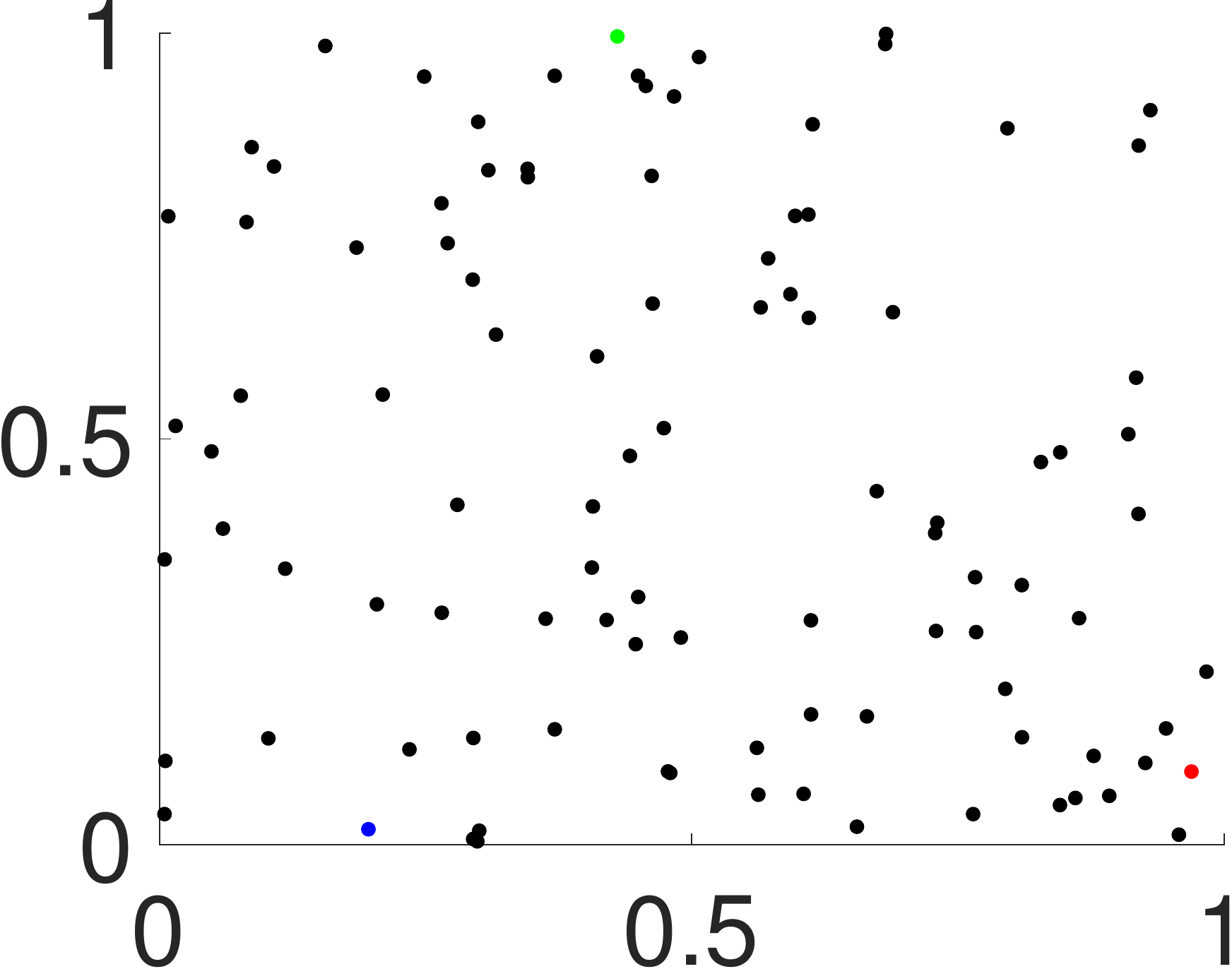} }
\subfigure[]{\includegraphics[width=0.45 \linewidth]{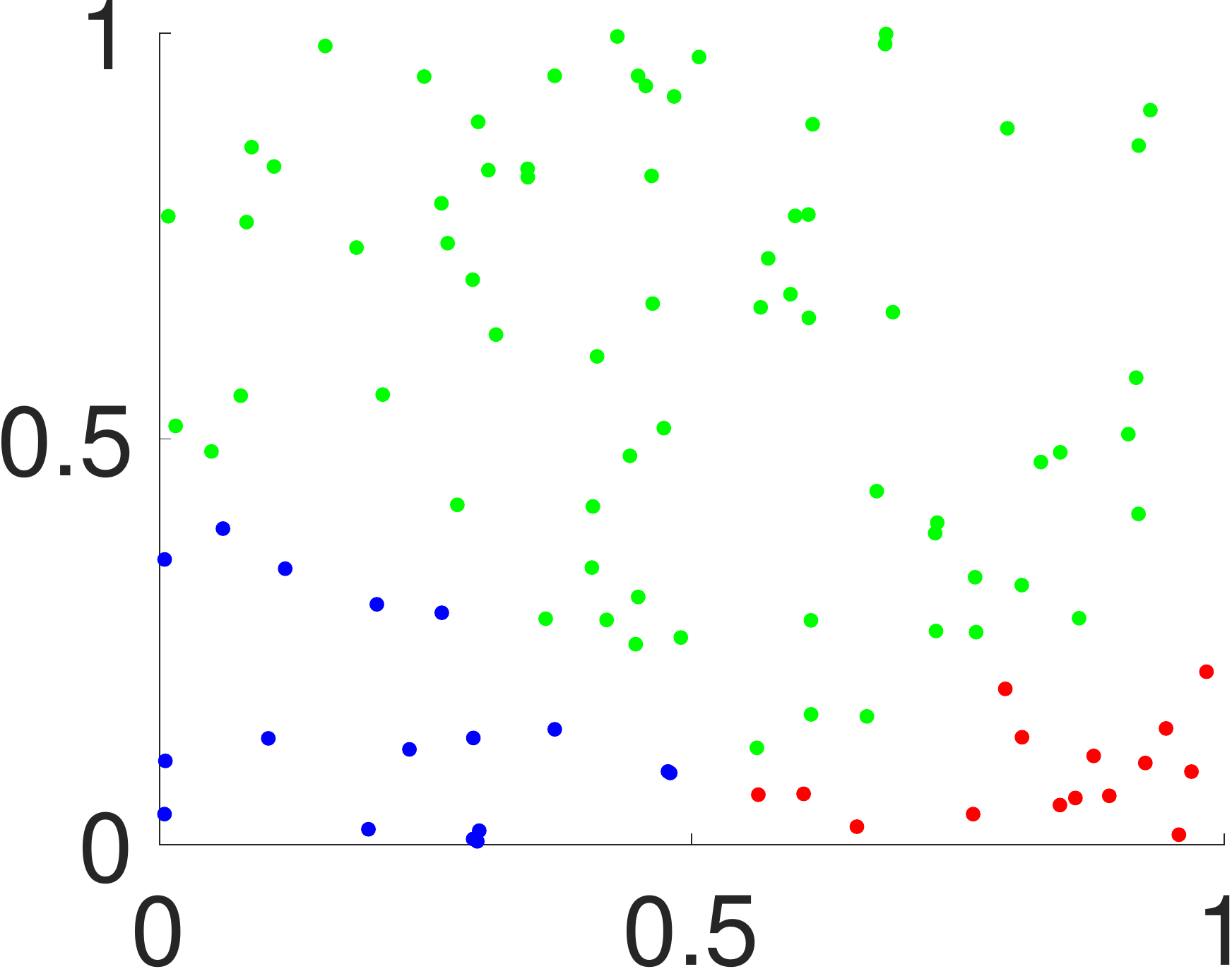} }\\
\subfigure[]{\includegraphics[width=0.45 \linewidth]{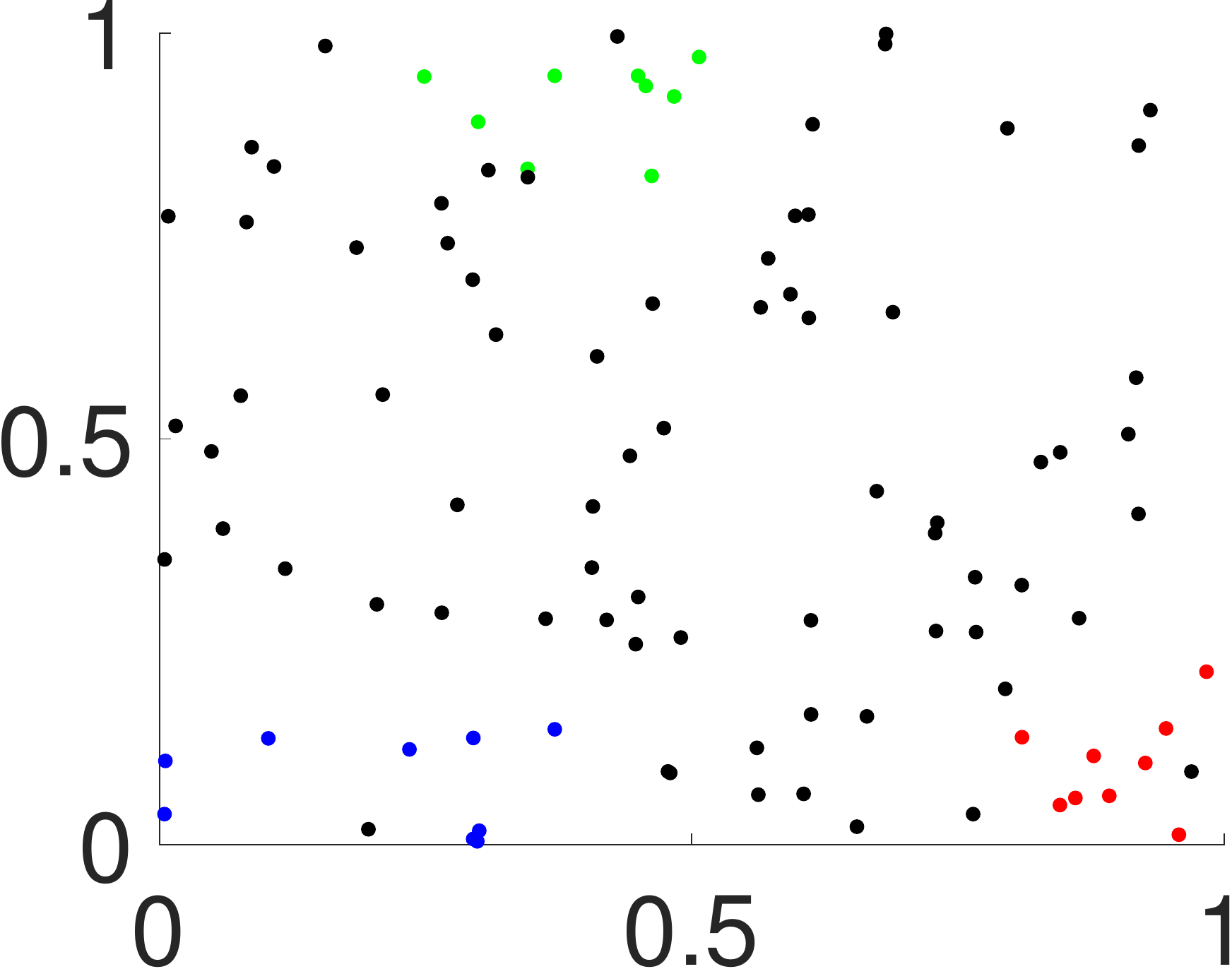} }
\subfigure[]{\includegraphics[width=0.45 \linewidth]{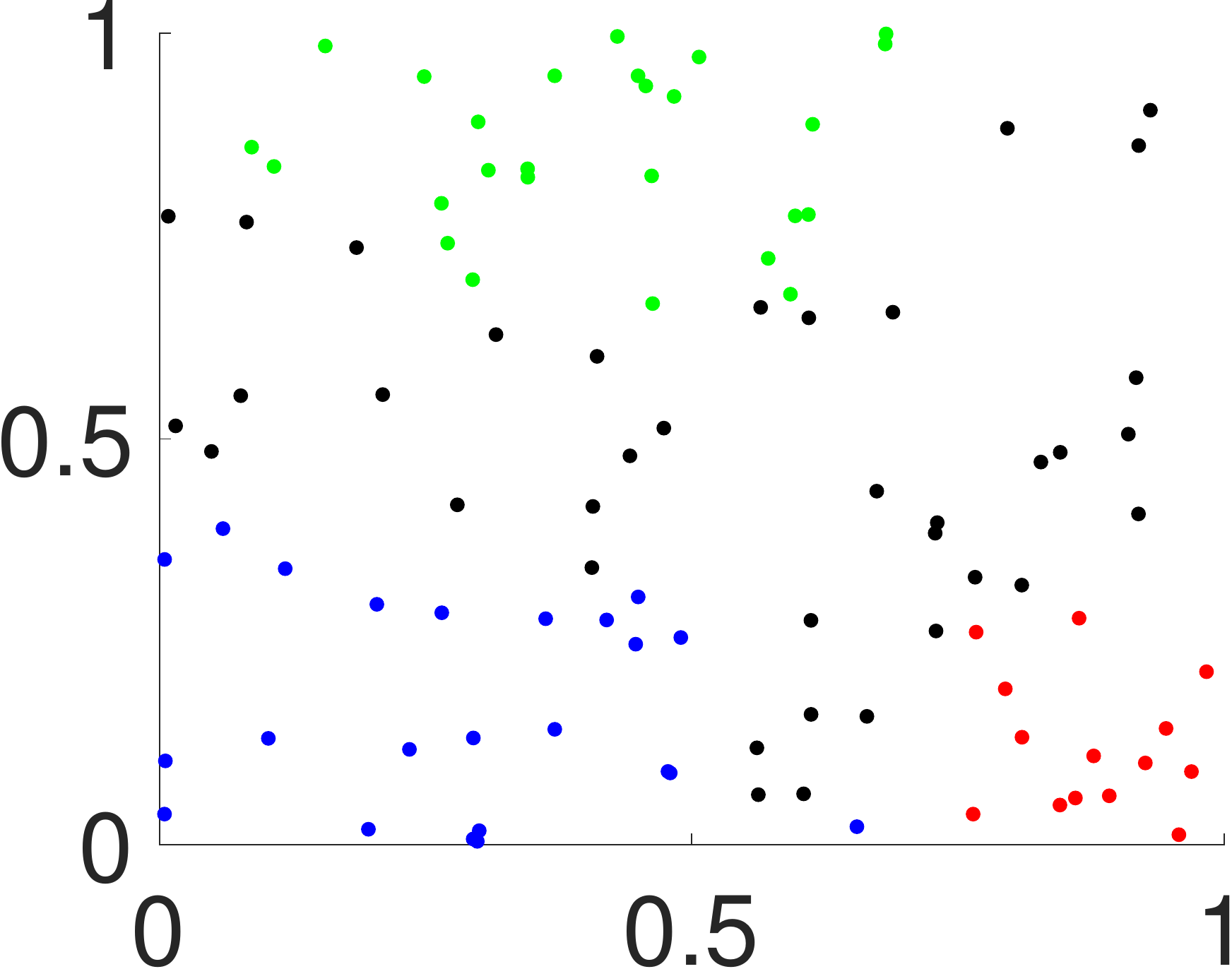} }\\
\subfigure[]{\includegraphics[width=0.45 \linewidth]{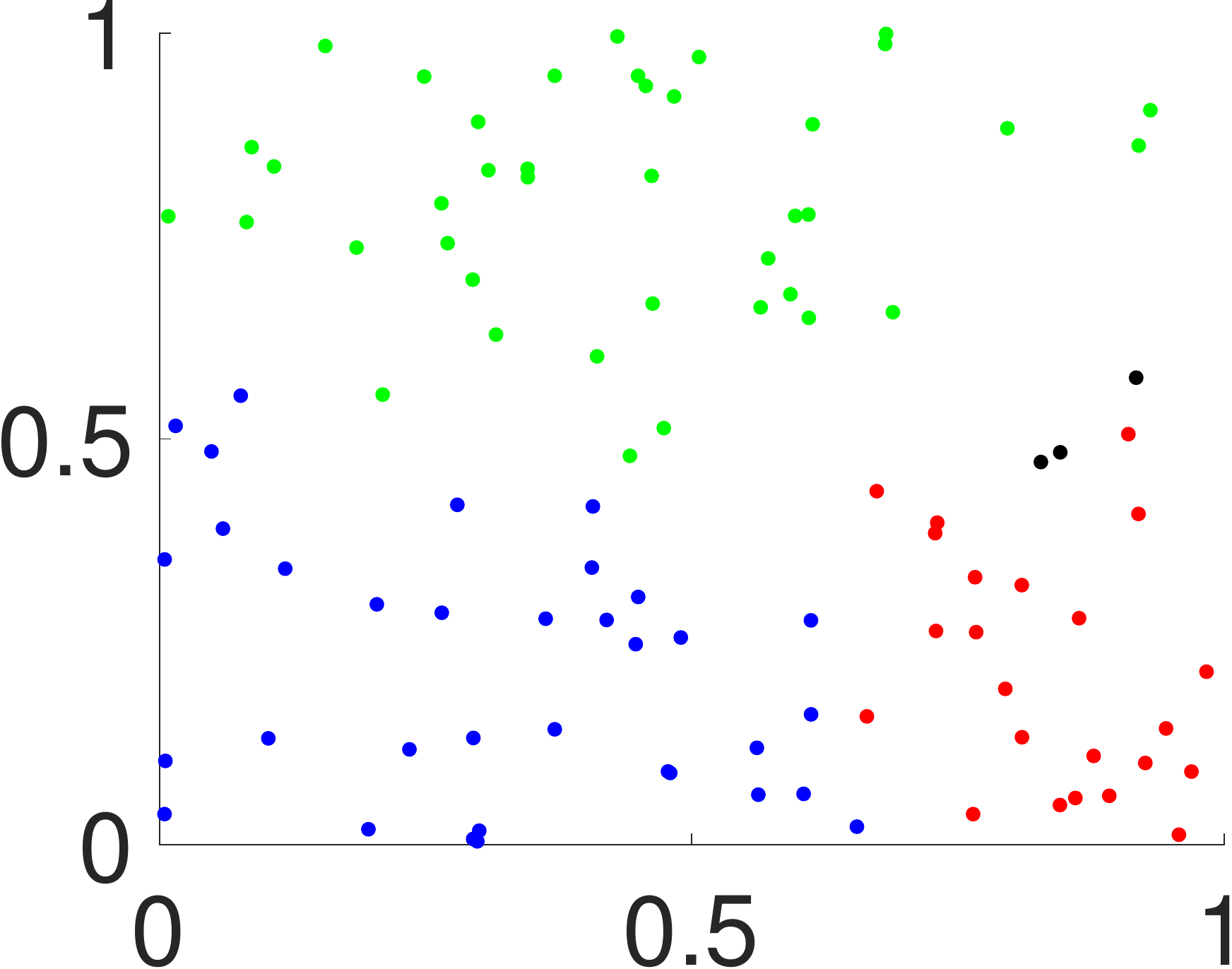} }
\subfigure[]{\includegraphics[width=0.45 \linewidth]{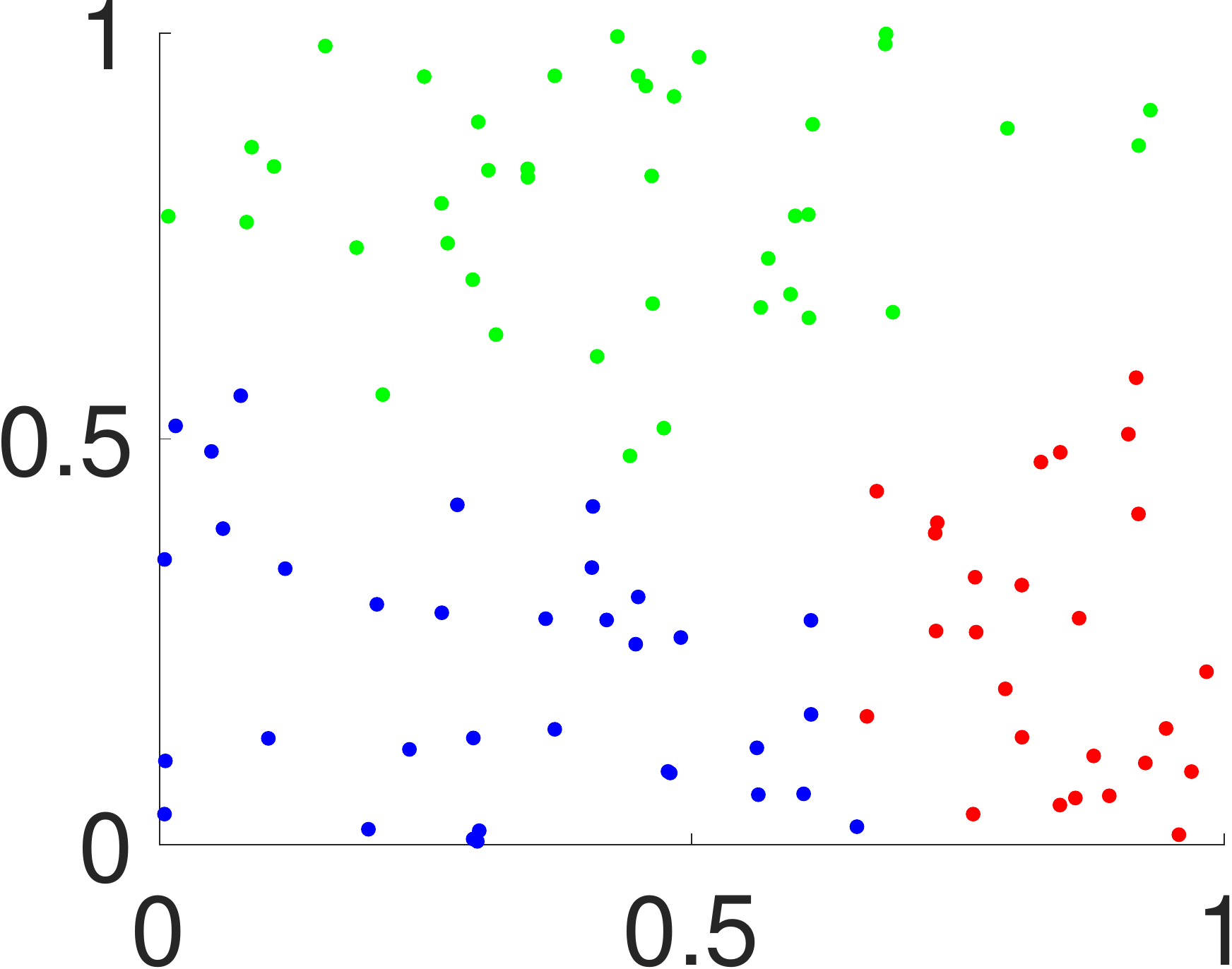} }
\caption{Graph signal interpolation. (a) Initial samples for $3$ clusters. The graph signal is indicated by color. (b) Solution of~\eqref{equ:sparseOpt}. (c) Result of the first iteration. The set $\mathcal{M}_1$ is marked by the color of the interpolated graph signal. The remaining nodes are marked in black (d) Result of the second iteration. (e) Result of the third iteration.  (f) Final result.}
\label{fig:toy_example}
\end{figure}

\textcolor{blue}{}{
\subsection{Convergence of Iterative Interpolation}
Algorithm~\ref{alg:it} will converge as long as each connected component in $\mathcal{G}$ contains a node in the sampling set $\mathcal{M}_1$. Specifically, we denote by $b$ the minimal integer for which the $b$-hop neighborhood of each node  contains all the nodes in the graph. Such a $b$ is guaranteed to exist since there exists a path between each node in the graph and a node in the sampling set $\mathcal{M}_0$. It follows that  $\mathcal{M}_{b+1} = \mathcal{V}$. }

Consequently,  the solution $\mathbf{\hat{s}}_b$ (to the $b$th iteration of Algorithm~\ref{alg:it}), is  determined by a system of $N$ equations and $N$ unknowns. The eigenvalues of $\mathbf{P}$ are known to decrease to zero.  As the eigenvectors of $\mathbf{P}$ are linearly independent, there exists a single solution $\mathbf{\hat{s}}_b$. In all subsequent iterations, the set $\mathcal{M}_c = \mathcal{V}$. This leads to the same system of equations and, as a consequence, the same solution. Therefore, Algorithm~\ref{alg:it} has converged.

\textcolor{blue}{}{
If the graph $\mathcal{G}$ contains a connected component that is not represented in the sampling set, the set $\mathcal{M}$ would still remain unchanged after $b$ iterations. There is therefore no point in running the algorithm further.
}

\section{Nystr\"{o}m Graph Signal Interpolation}
\label{subsec:bigData}

The solution of~\eqref{equ:sparseOpt}
necessitates computation of the eigendecomposition of the graph shift operator. This 
computation is costly in terms of both complexity and memory consumption. When the graph has many nodes it may not 
be feasible to compute eigenvectors and eigenvalues of the graph shift operator. However, when the 
graph shift operator is the Markov matrix, a variation on the 
Nystr\"{o}m extension  \cite{fowlkes2004nystrom, nystrom1928nystron, baker1977nystrom, press1992numerical} can be used for semi-supervised learning of big data.

\subsection{The Nystr\"{o}m Extension}

We begin by giving a short introduction to the Nystr\"{o}m extension. 
Let $\{\mathbf{x}_i\}_{i=1}^{N}$ be a set of data points. A matrix $\mathbf{K} \in \mathbb{R}^{N \times N}$ 
is constructed such that $K_{i,j}=k\left( \mathbf{x}_i, \mathbf{x}_j \right)$, where $k \left( \cdot \right)$ is 
some kernel function and $\mathbf{K}$ is a positive semi-definite (PSD) matrix.  
The matrix $\mathbf{K}$ can be 
considered as a combination of four block matrices,
\begin{equation} \label{equ:structure}
\mathbf{K} = \begin{bmatrix}
\mathbf{E}\quad \mathbf{B}^T  \\
\mathbf{B} \quad \mathbf{C}
\end{bmatrix},
\end{equation}
where $\mathbf{E}  \in \mathbb{R}^{r \times r}$, $\mathbf{ B}  \in \mathbb{R}^{N-r \times r}$, and $\mathbf{C} 
\in  \mathbb{R}^{N-r \times N-r}$ for some $0 < r < N$. 
The Nystr\"{o}m extension is a method for extending the eigenvectors of $\mathbf{E}$ to create an estimate of $r$  
eigenvectors of $\mathbf{K}$. 

Let $\mathbf{Z} \in \mathbb{R}^{r \times r}$ be the matrix whose columns are the eigenvectors of $\mathbf{E}$. 
As $\mathbf{E}$ contains the first $r$ rows and the first $r$ columns of $\mathbf{K}$, it is itself a symmetric matrix, thus
\begin{equation} \label{equ:diagA}
\mathbf{E} = \mathbf{Z} \mathbf{Q} \mathbf{Z}^T
\end{equation}
where $\mathbf{Q}$ is a diagonal matrix containing the eigenvalues of $\mathbf{E}$. 
The Nystr\"{o}m  extension of the matrix of eigenvectors of $\mathbf{K}$ is given by 
\begin{equation} \label{equ:Nystrom}
\tilde{\mathbf{Z}} = \begin{bmatrix} \mathbf{Z} \\
\mathbf{B} \mathbf{Z} \mathbf{Q}^{-1}
\end{bmatrix}.
\end{equation}
For more details  see \cite{fowlkes2004nystrom}.

We suggest that when computing $\mathbf{B Z Q}^{-1}$ all eigenvalues be approximated as ones, 
resulting in the following modification to~\eqref{equ:Nystrom}
\begin{equation} \label{equ:revised}
\tilde{\mathbf{Z}} = \begin{bmatrix} \mathbf{Z} \\
\mathbf{B} \mathbf{Z}
\end{bmatrix}.
\end{equation}
We motivate this approximation in  Appendix C.

\subsection{Nystr\"{o}m Interpolation}

The application of the Nystr\"{o}m extension to the Markov matrix 
is not straightforward. This is due to the fact that 
the Nystr\"{o}m extension is geared towards PSD matrices, and  
the Markov matrix $\mathbf{P}$ is not  PSD. However, as we show in Proposition~\ref{prop:1} (see Appendix A), 
the 
Markov matrix is 
strongly related to the normalized graph Laplacian $\mathbf{L}$ which is 
PSD, and defined by
\begin{equation}\label{equ:laplacian}
  \mathbf{L} = \mathbf{D}^{-\frac{1}{2}} \left(\mathbf{D} - \mathbf{W} \right)
  \mathbf{D}^{-\frac{1}{2}} = \mathbf{I}_N - \mathbf{D}^{-\frac{1}{2}}
  \mathbf{W} \mathbf{D}^{-\frac{1}{2}}
\end{equation}
where $\mathbf{I}_N$ is the $N \times N$ identity matrix. It is easy to see that 
  $\mathbf{D}^{-\frac{1}{2}} \mathbf{L} \mathbf{D}^{\frac{1}{2}} =
  \mathbf{I}_N - \mathbf{P}$.
Thus, $\mathbf{L}$ is similar to $\mathbf{I}_N - \mathbf{P}$. From Proposition~\ref{prop:1} (see Appendix A), 
the connection between the matrix of eigenvectors of the Markov matrix $\mathbf{V}$ 
and the matrix of eigenvectors of the Laplacian $\mathbf{U}$ is,  
\begin{equation} \label{equ:eigV1}
\tilde{\mathbf{V}} = \mathbf{D}^{-\frac{1}{2}}\tilde{\mathbf{Z}},
\end{equation}
and the matrix approximating its eigenvalues is
 \begin{equation} \label{equ:L}
 \tilde{\mathbf{\Lambda}} = \mathbf{I}_N - \mathbf{Q}.
  \end{equation}

To obtain an efficient graph signal interpolation algorithm we insert~\eqref{equ:eigV1} and~\eqref{equ:L} into~\eqref{equ:sparseOpt}. 
Specifically, $\psi_i$ is replaced by the $i$th column of $\tilde{\mathbf{V}}$, and $\lambda_i$ is replaced by $ \tilde{\mathbf{\Lambda}} _{i,i}$.

The difference between our two interpolation methods  
is only in the computation of the eigendecomposition of the Markov matrix. In our smoothness interpolation, presented in Section~\ref{sec:mainIdea}, 
the eigenvectors and eigenvalues of the full $N \times N$ matrix must be determined. 
On the other hand,  Nystr\"{o}m smoothness interpolation  uses the   eigendecomposition 
of an $r \times r$ matrix to approximate the 
eigenvectors and eigenvalues of the $N \times N$ matrix. When $r \ll N$ this method is \textcolor{blue}{}{computationally} efficient. We show in  Section \ref{subsec:Experiment} that 
this approach still achieves good accuracy in simulations on the MNIST dataset \cite{lecun1998mnist}.  We note that Nystr\"{o}m smoothness interpolation 
can also be done iteratively, as detailed in Algorithm~\ref{alg:it}. 

\textcolor{blue}{}{The Nystr\"{o}m extension is not meant to be used on graphs with sparse graph shift matrices (that is, few edges). Rather,  this method is geared towards dense graph shift matrices, in which case few samples are needed.  
We note that computationally efficient methods for graph signal interpolation over sparsely connected graphs through message passing have been suggested~\cite{jung2019message}. These methods are not computationally efficient for more densely connected graphs.}

\section{Experimental Results for Graph Signal Interpolation}
\label{subsec:Experiment}

\subsection{Synthetic Simulations}

In this section we present a qualitative comparison between our framework and those  of~\cite{keller2011regression, chen2016reconstruction, chen2015var, ma2015jmlr} for several signals and several sampling strategies. We use the Graph Signal Processing Toolbox \cite{perraudin2014gspbox} to produce the bunny graph depicted in Fig.~\ref{fig:bunny}. This graph contains $2503$ nodes and $27452$ edges connecting nearby nodes. In each experiment we use the Markov matrix as the graph shift operator.

 We further note that, as~\eqref{equ:sparseOpt} is an $l_1$ optimization problem, we use the SPGL1 package\footnote{https://github.com/mpf/spgl1.} \cite{vanderberg2008spg, vanderberg2010spg} 
to solve it. 

\begin{figure}
  \centering
  \subfigure[]{\includegraphics[width=0.3\linewidth]{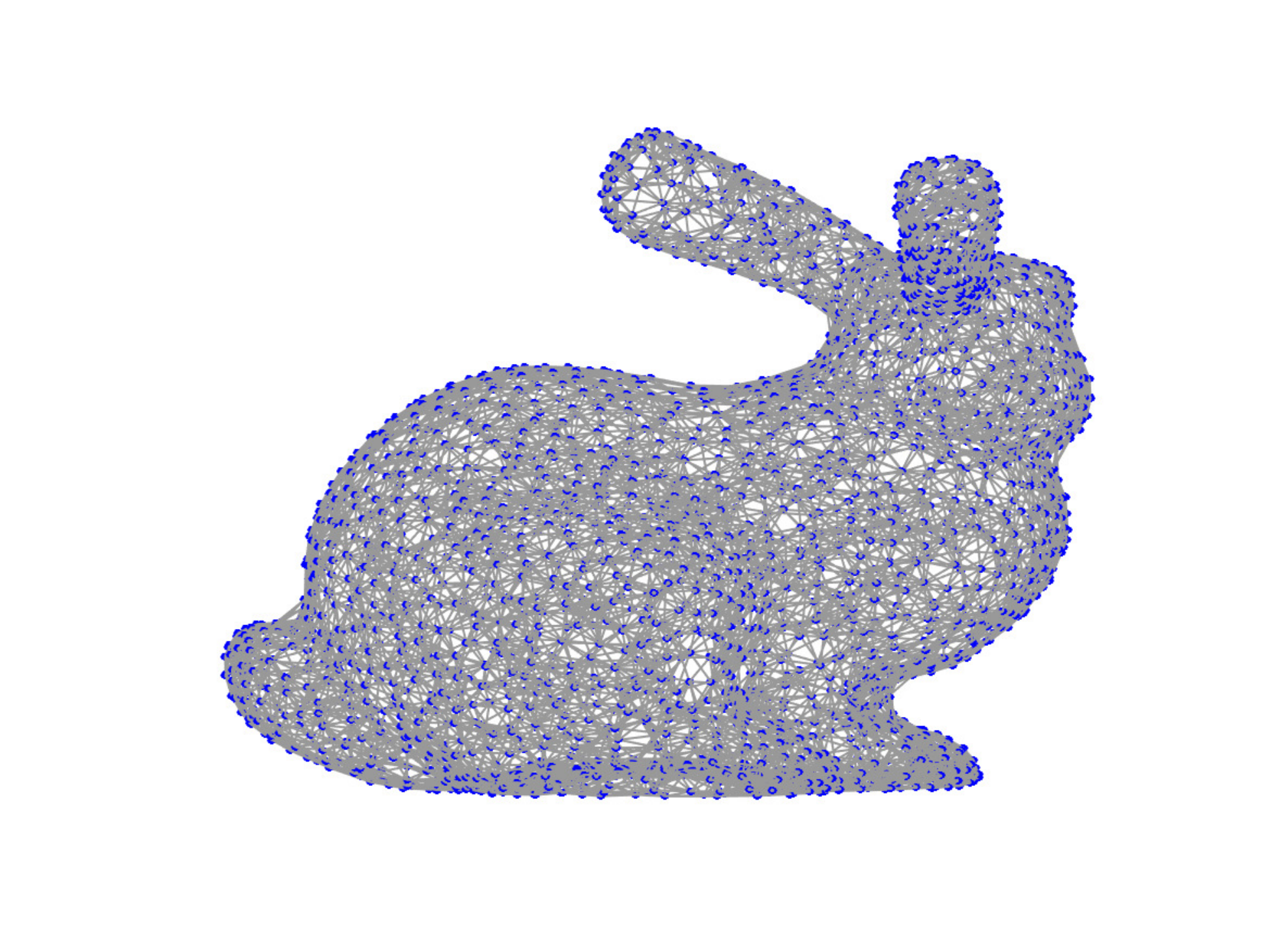}}
   \subfigure[]{\includegraphics[width=0.3\linewidth]{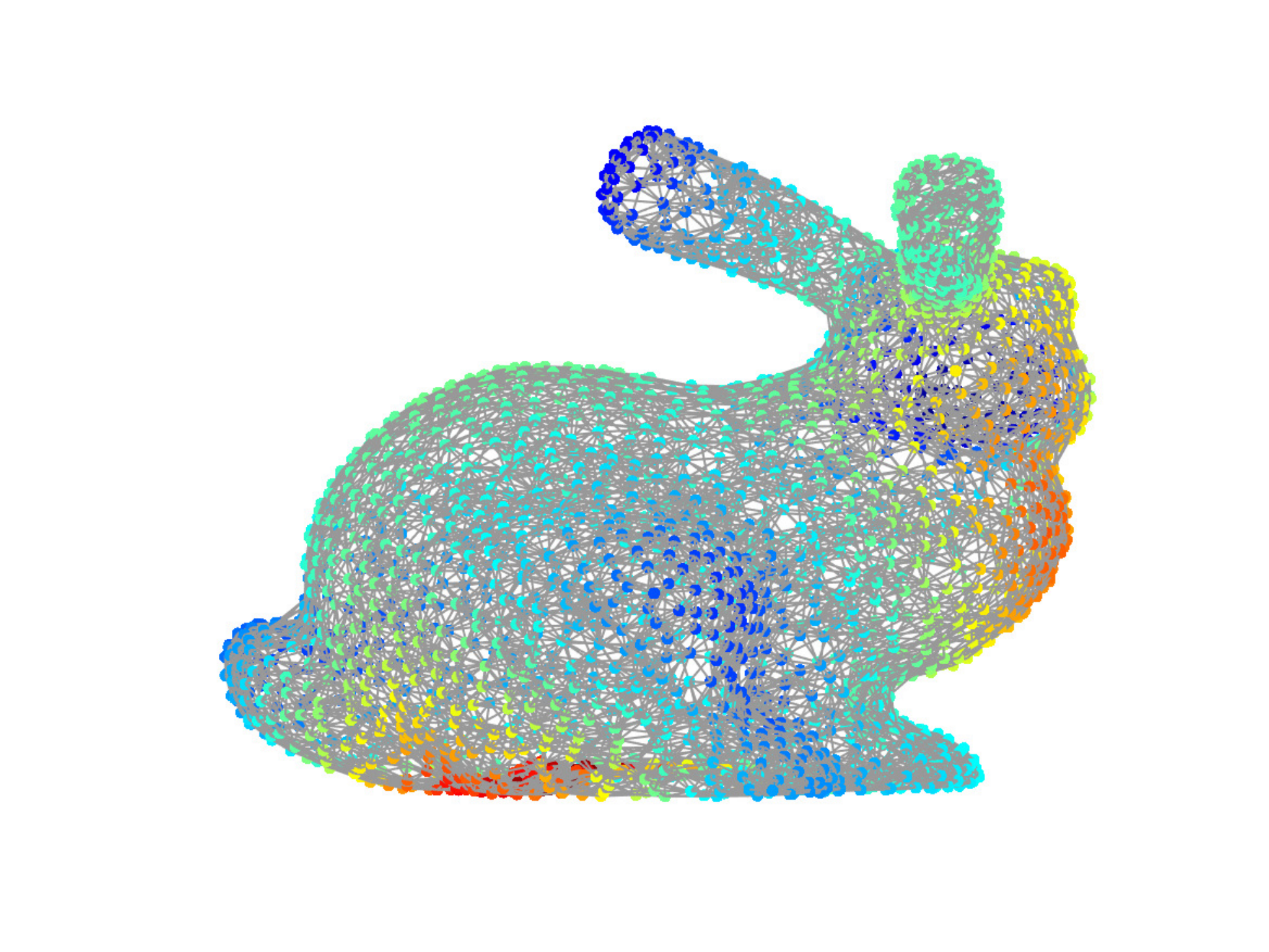} }
      \subfigure[]{\includegraphics[width=0.3\linewidth]{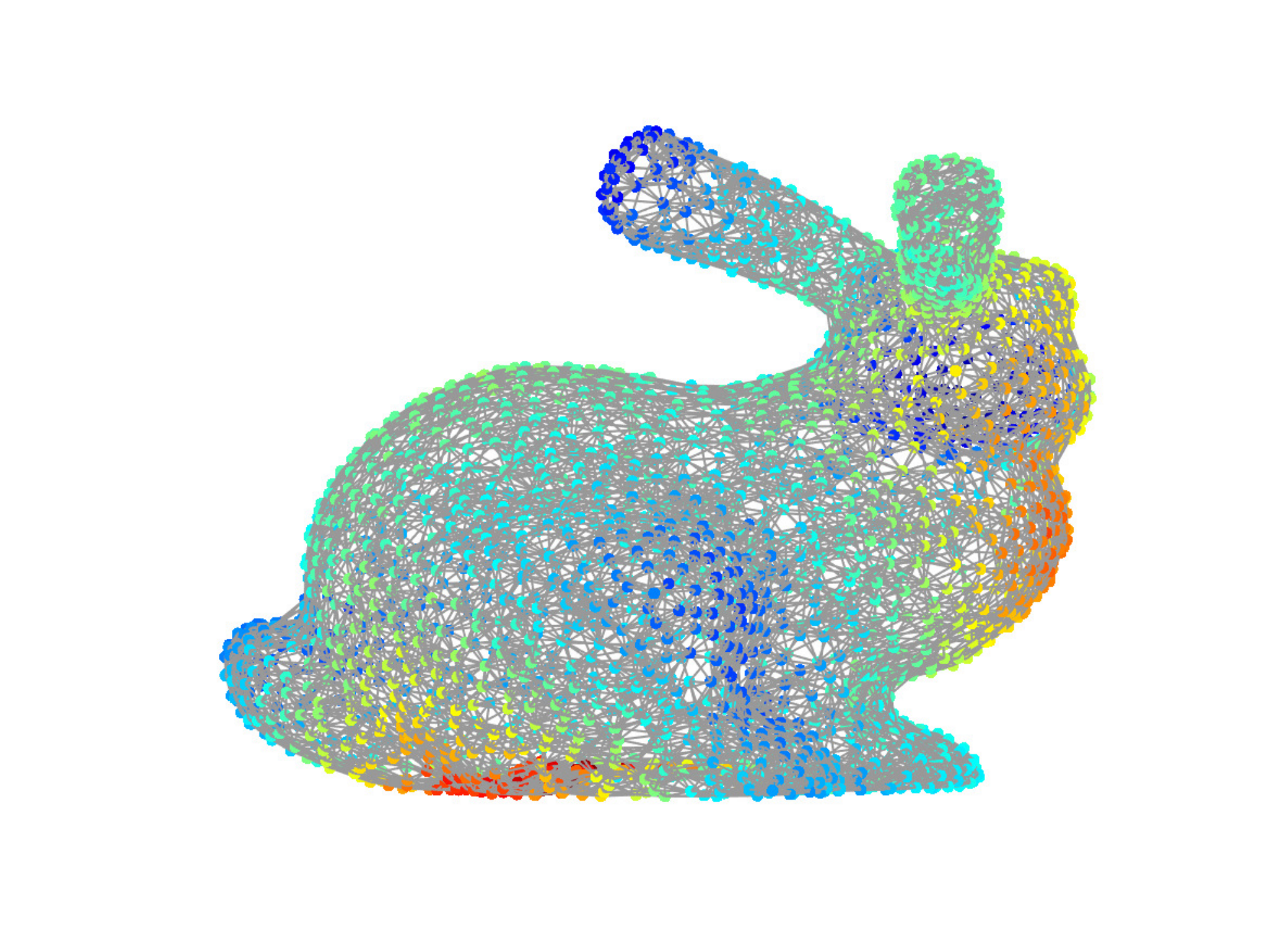} }
  \caption{Bunny graph as created by the Graph Signal Processing Toolbox \cite{perraudin2014gspbox}. (a) Nodes and edges of the graph. (b) Bandlimited graph signal. The color of each node corresponds to its signal. (c) Approximately bandlimited graph signal.}
  \label{fig:bunny}
\end{figure}

\subsubsection{Bandlimited graph signal}

We created a $20$-bandlimited graph signal. This was done by randomly generating the first twenty entries of the spectrum of the graph signal $\hat{\mathbf{x}}$ and setting the remaining entries to zero. The resulting graph signal is $$\mathbf{x} = \mathbf{V} \hat{\mathbf{x}},$$ and is depicted in Fig.~\ref{fig:bunny}. 

Next, we sample the nodes in the graph. We use two sampling strategies. The first is random sampling, where each node has a uniform probability of being sampled. Additionally we use the sampling procedure suggested in~\cite{chen2015sampling} (which, in this case, allows for perfect reconstruction). We compare the error of~\eqref{equ:sparseOpt} with the error of spectral regression, noiseless inpainting~\cite{chen2015var} and both reconstruction strategies discussed in~\cite{chen2016reconstruction} (where one of these was introduced in~\cite{ma2015jmlr}). We compute the error of each reconstruction method for varying sizes of sample sets. The error is defined as 
\begin{equation}
\left\lVert \frac{\mathbf{y}}{\Vert \mathbf{y} \Vert_2} - \frac{\mathbf{\hat{y}}} {\Vert \mathbf{\hat{y}} \Vert} \right\rVert ,
\label{equ:error1}
\end{equation} 
where the true signal is $\mathbf{y}$ and the estimate is denoted by $\mathbf{\hat{y}}$.

In Fig.~\ref{fig:bandlimit_bunny_k} we present interpolation results when the bandlimit is known (for all methods except noiseless inpainting~\cite{chen2015var} as this solution does not refer to the eigenvectors of the graph shift). In Fig.~\ref{fig:bandlimit_bunny_uk} we plot interpolation results when the bandlimit is unknown and assumed to be unlimited. 

We note that in the case of the optimal sampling operator~\cite{chen2015sampling}, our method is the only one that recovers the graph signal exactly both when the bandlimit is known and when the bandlimit is unknown. The reason for this is that in our method the bias towards smooth graph signals does not depend on any predetermined bandlimit.

\begin{figure}
  \centering
 \subfigure[]{\includegraphics[width=0.4\linewidth]{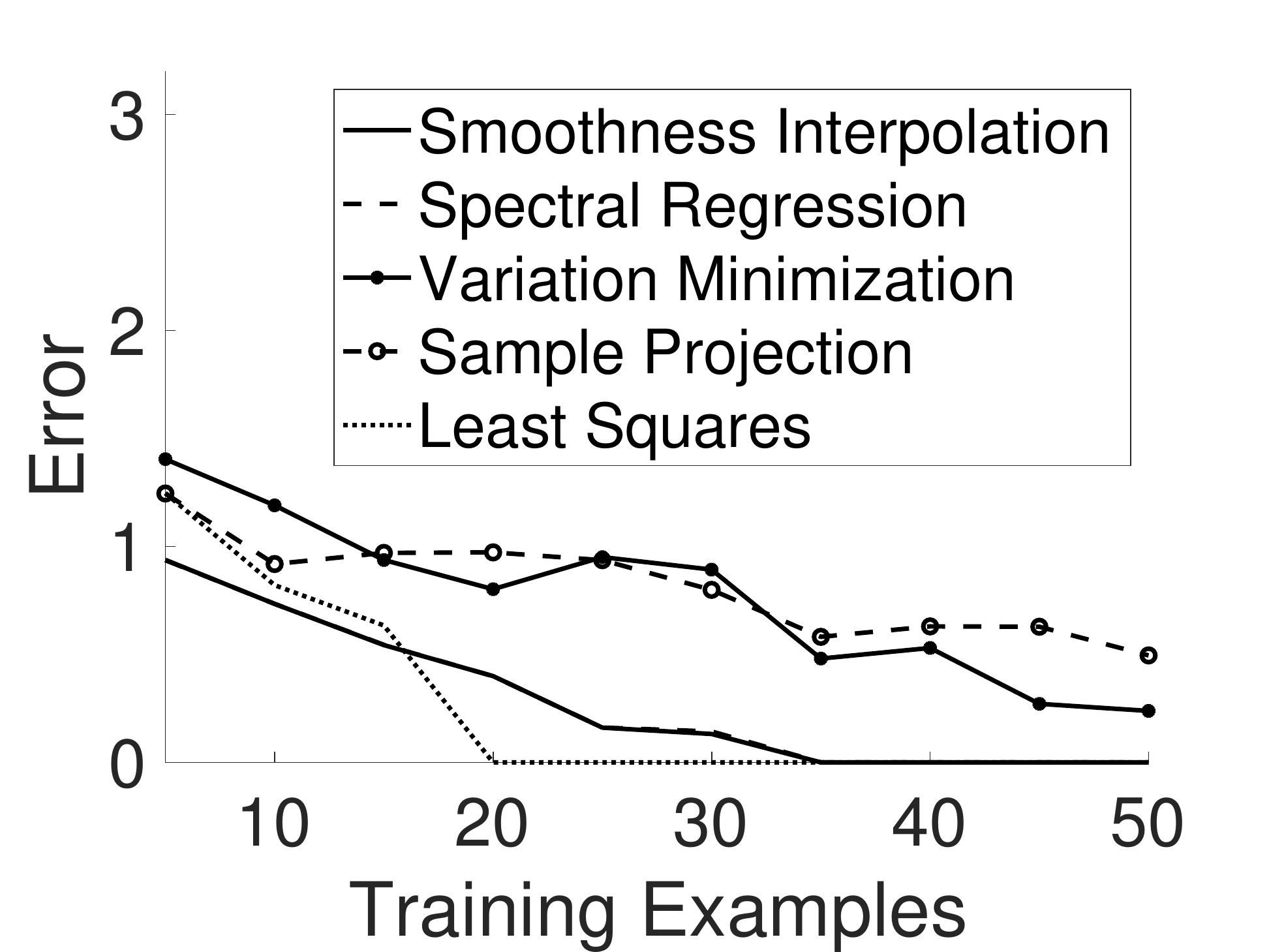}}
 \subfigure[]{\includegraphics[width=0.4\linewidth]{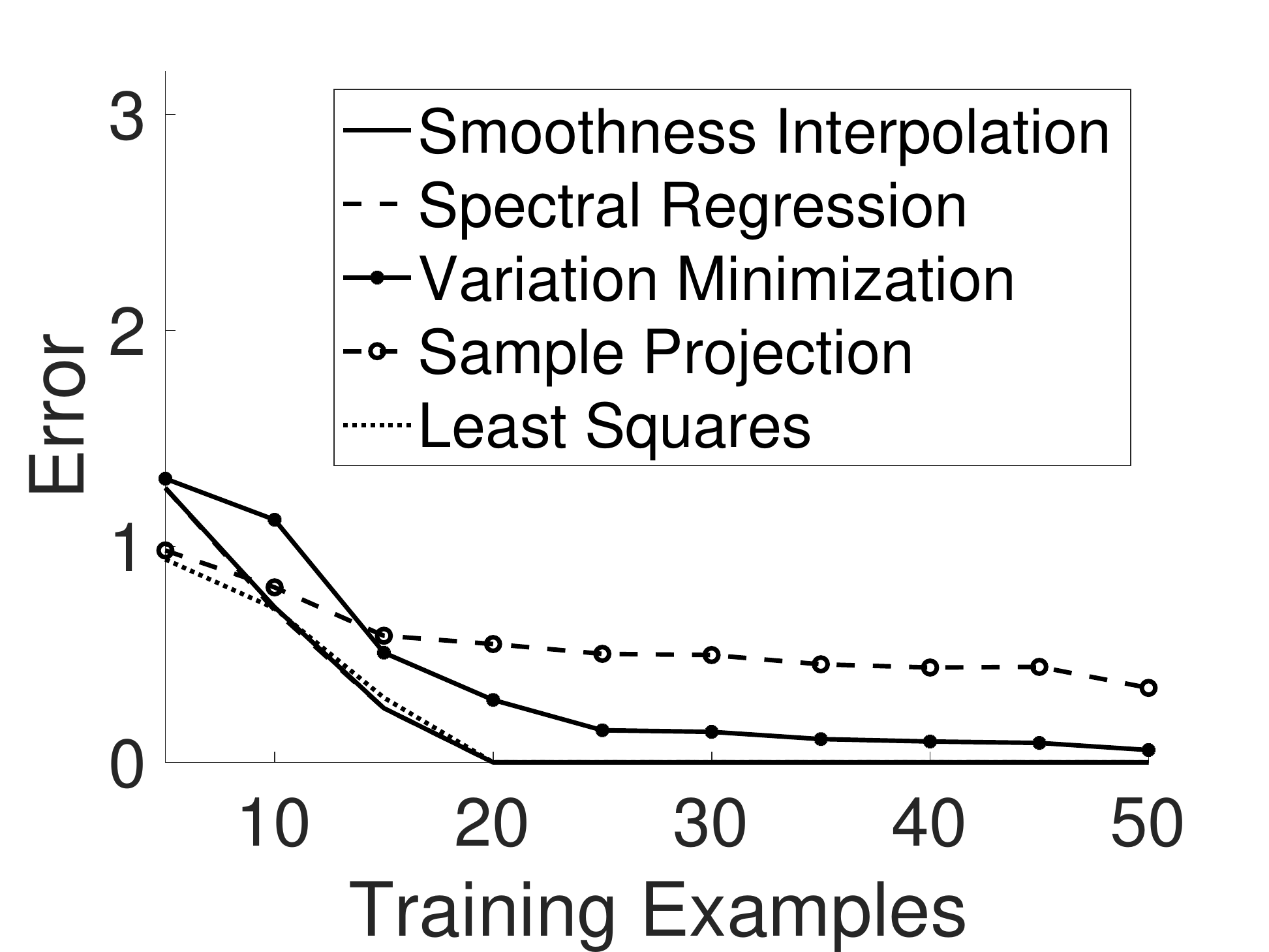}}
  \caption{Error rate~\eqref{equ:error1} for reconstruction of $20$-bandlimited graph signal when bandlimit is known. The result of~\eqref{equ:sparseOpt} is presented in blue. The results of spectral regression,~\cite{chen2015var},~\cite{chen2016reconstruction} and~\cite{ma2015jmlr}  are presented in red, yellow, purple and green, respectively. (a) Results for sampling set selected uniformly at random. (b) Results for sampling set selected as detailed in~\cite{chen2015sampling}.}
\label{fig:bandlimit_bunny_k}  
\end{figure}

\begin{figure}
  \centering
 \subfigure[]{\includegraphics[width=0.4\linewidth]{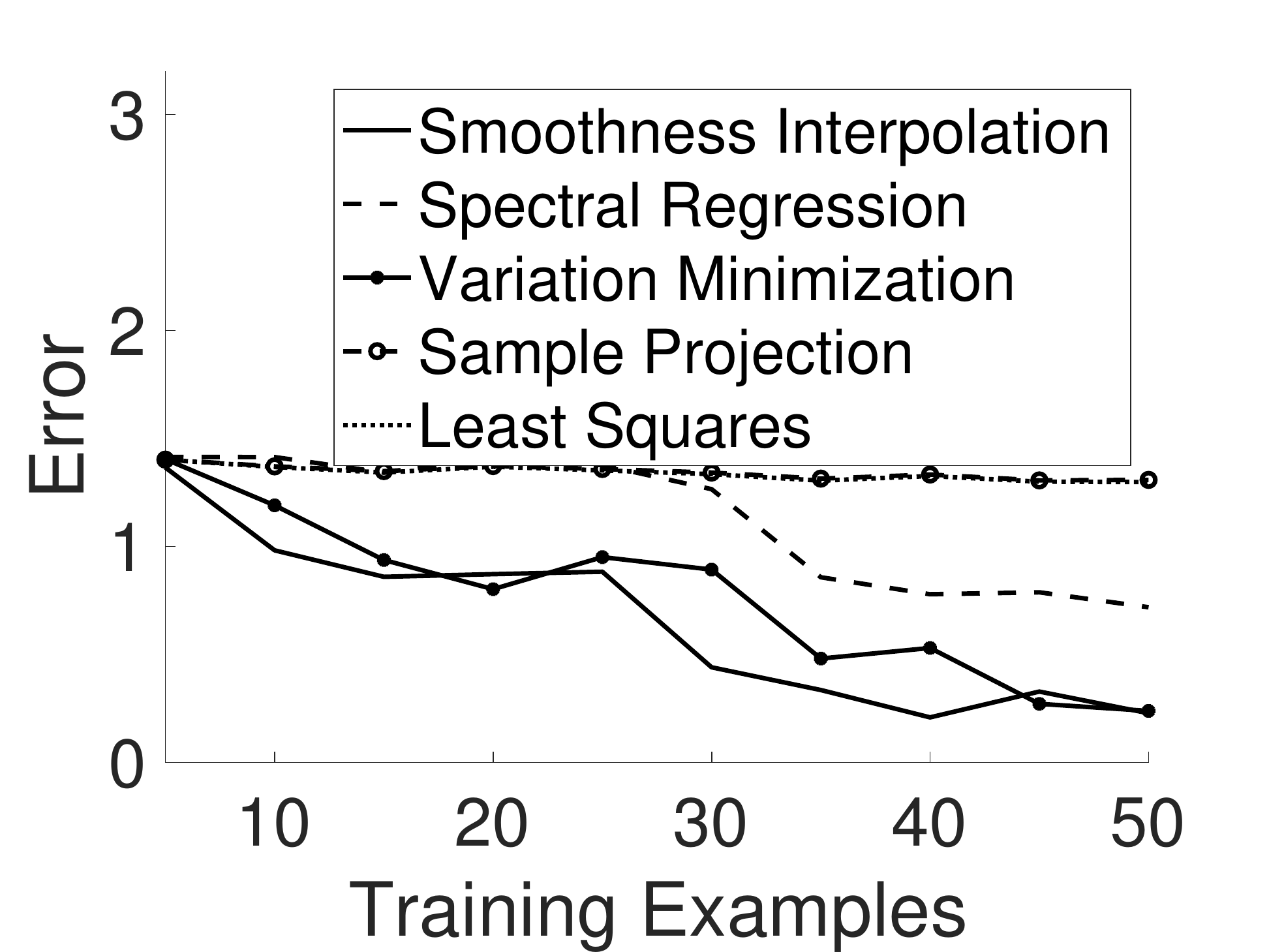}}
 \subfigure[]{\includegraphics[width=0.4\linewidth]{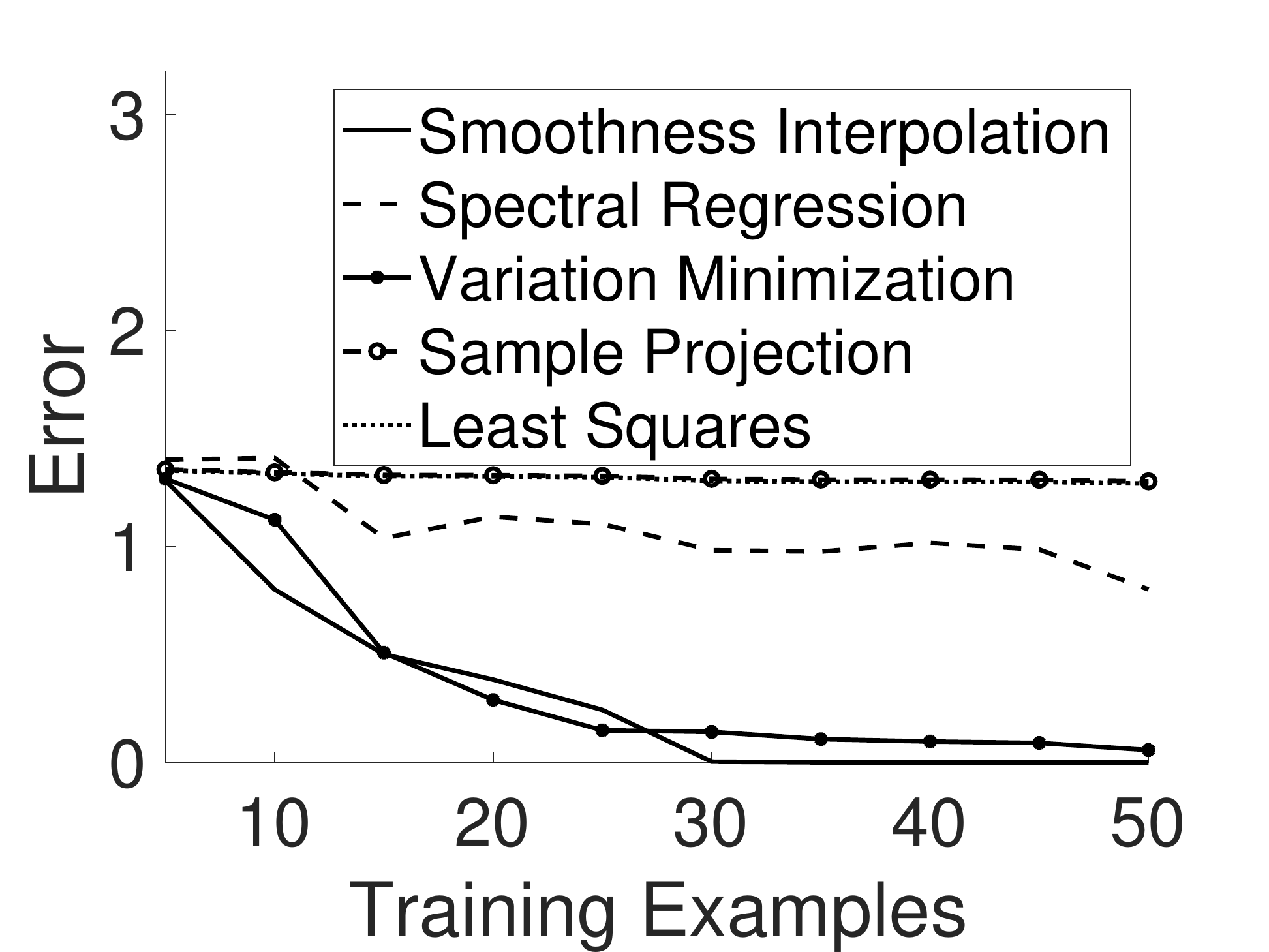}}
  \caption{Error rate~\eqref{equ:error1} for reconstruction of $20$-bandlimited graph signal when bandlimit is unknown. The result of~\eqref{equ:sparseOpt} is presented in blue. The results of spectral regression,~\cite{chen2015var},~\cite{chen2016reconstruction} and ~\cite{ma2015jmlr}  are presented in red, yellow, purple and green, respectively. (a) Results for sampling set selected uniformly at random. (b) Results for sampling set selected as detailed in~\cite{chen2015sampling}.}
  \label{fig:bandlimit_bunny_uk}
\end{figure}

\subsubsection{Approximately bandlimited graph signals}

In reality, graph signals are often approximately bandlimited. We simulate such signals by adding a noise vector (generated uniformly at random) to the spectrum of our $20$-bandlimited  graph signal. We interpolate the graph signal under the assumption that the signal is $20$-bandlimited (Fig.~\ref{fig:approx_bandlimit_bunny_k}) and under the assumption that the signal is not bandlimited (Fig.~\ref{fig:approx_bandlimit_bunny_uk}). We note that, as case of bandlimited graph signals,~\cite{chen2015var} makes no assumption on the bandlimit. We conclude that our interpolation is successful especially when no assumptions are made on the bandlimit. This property stems from a built-in bias towards negligible values in the high frequencies of the interpolated graph signal.

\begin{figure}
  \centering
 \subfigure[]{\includegraphics[width=0.4\linewidth]{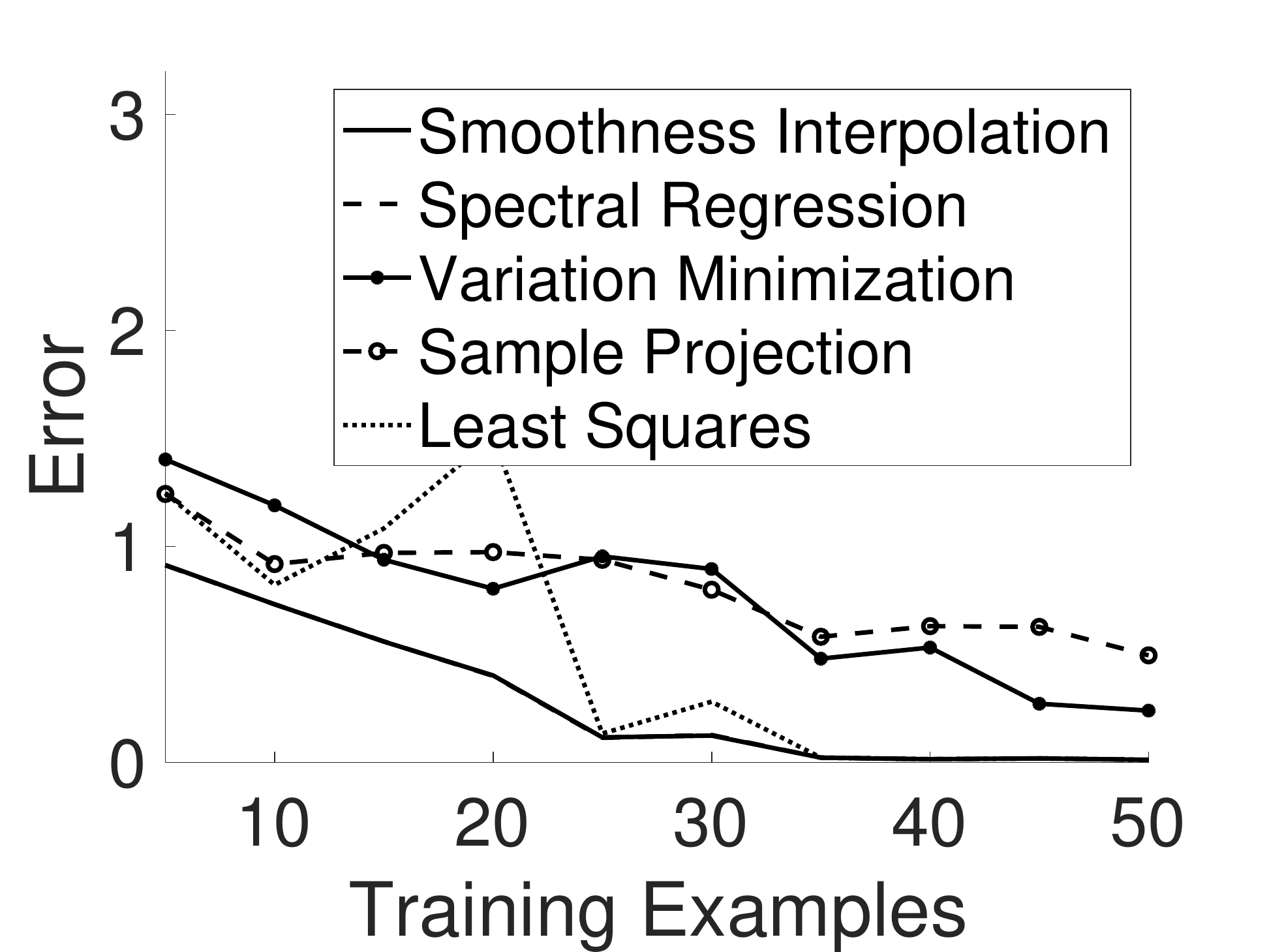}}
 \subfigure[]{\includegraphics[width=0.4\linewidth]{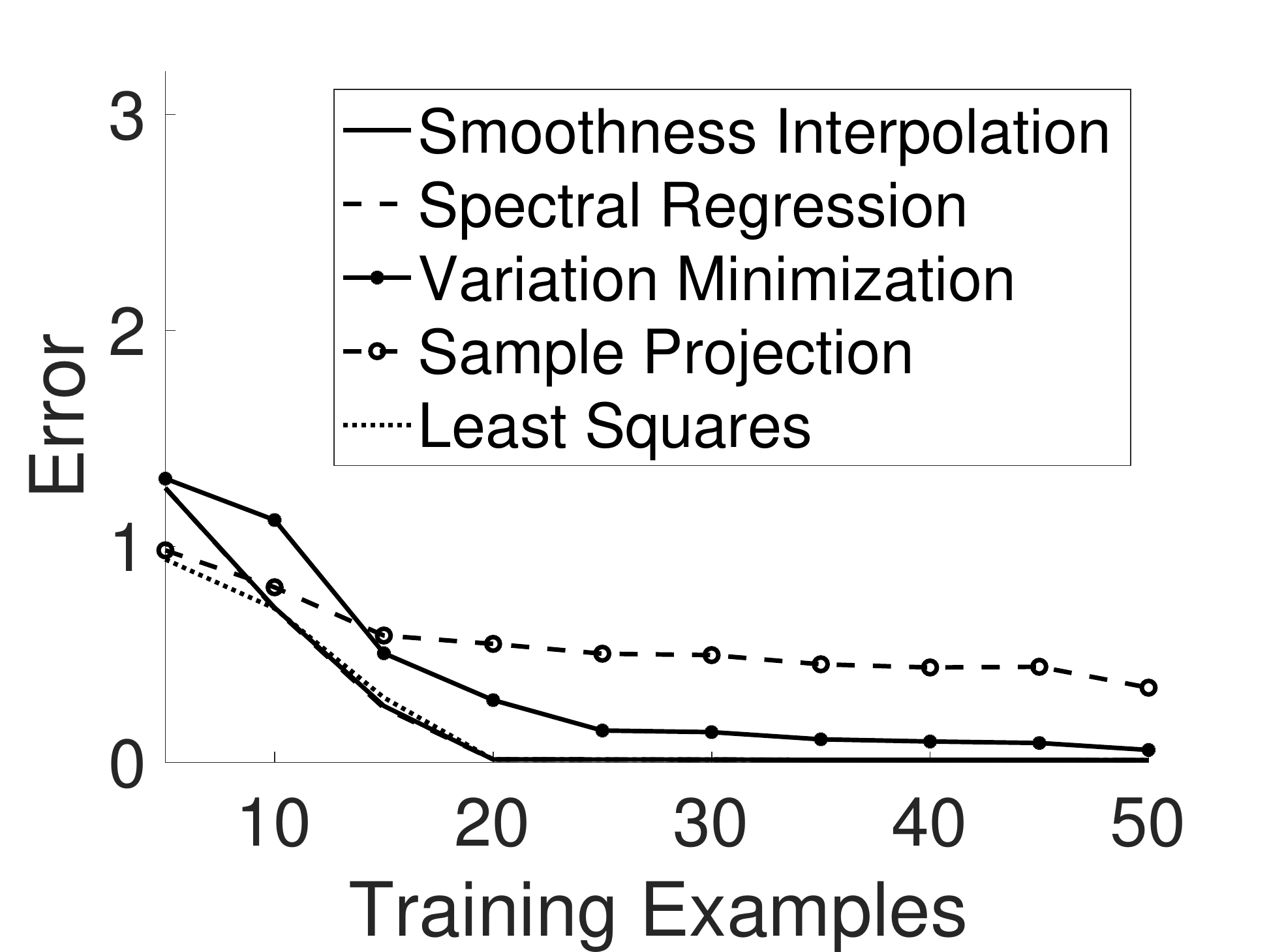}}
  \caption{Error rate~\eqref{equ:error1} for reconstruction of an approximately $20$-bandlimited graph signal. The bandlimit is assumed to be $20$. The result of~\eqref{equ:sparseOpt} is presented in blue. The results of spectral regression,~\cite{chen2015var},~\cite{chen2016reconstruction} and ~\cite{ma2015jmlr}  are presented in red, yellow, purple and green, respectively. (a) Results for sampling set selected uniformly at random. (b) Results for sampling set selected as detailed in~\cite{chen2015sampling}.}
  \label{fig:approx_bandlimit_bunny_k}
\end{figure}

\begin{figure}
  \centering
 \subfigure[]{\includegraphics[width=0.4\linewidth]{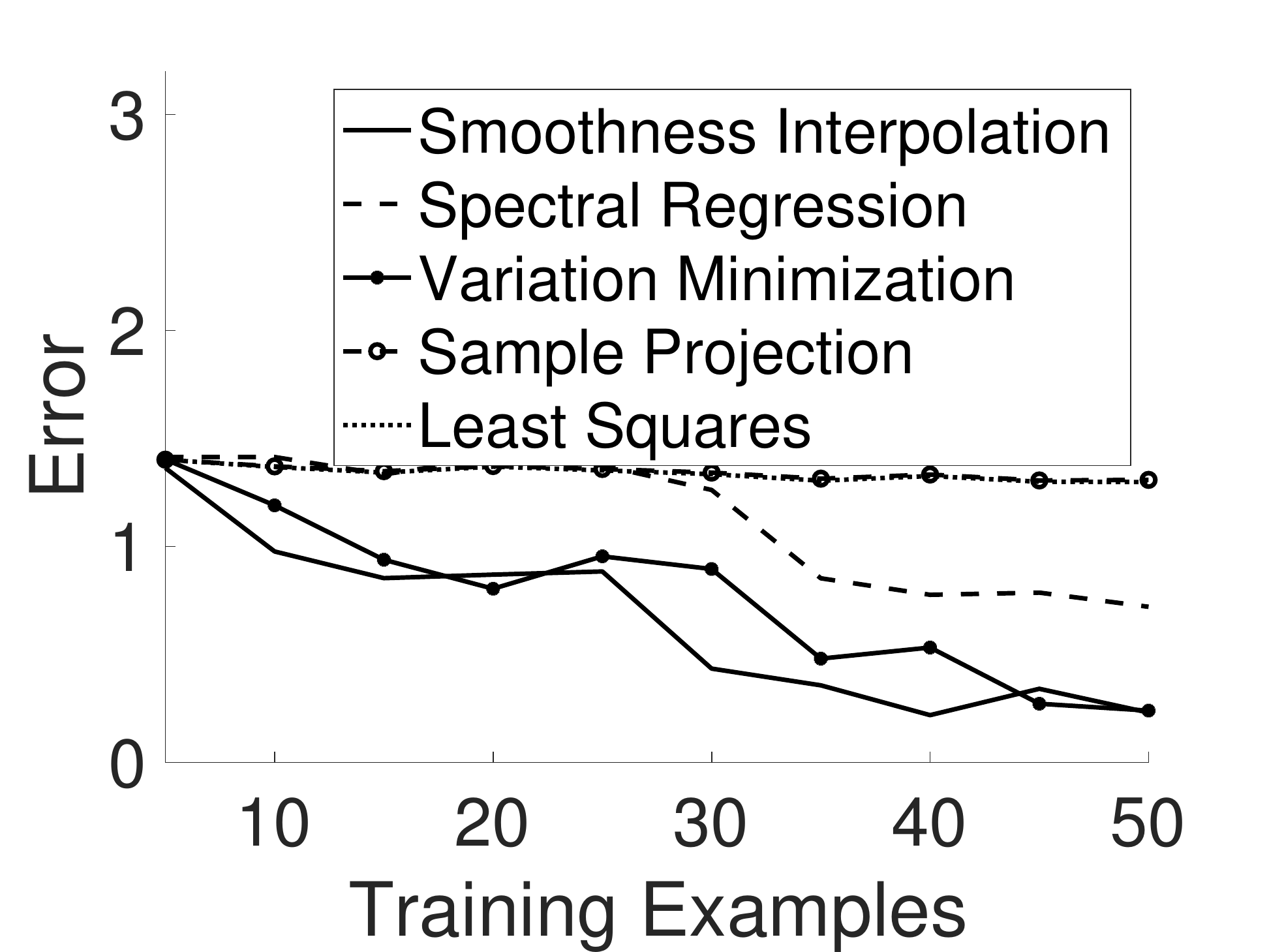}}
 \subfigure[]{\includegraphics[width=0.4\linewidth]{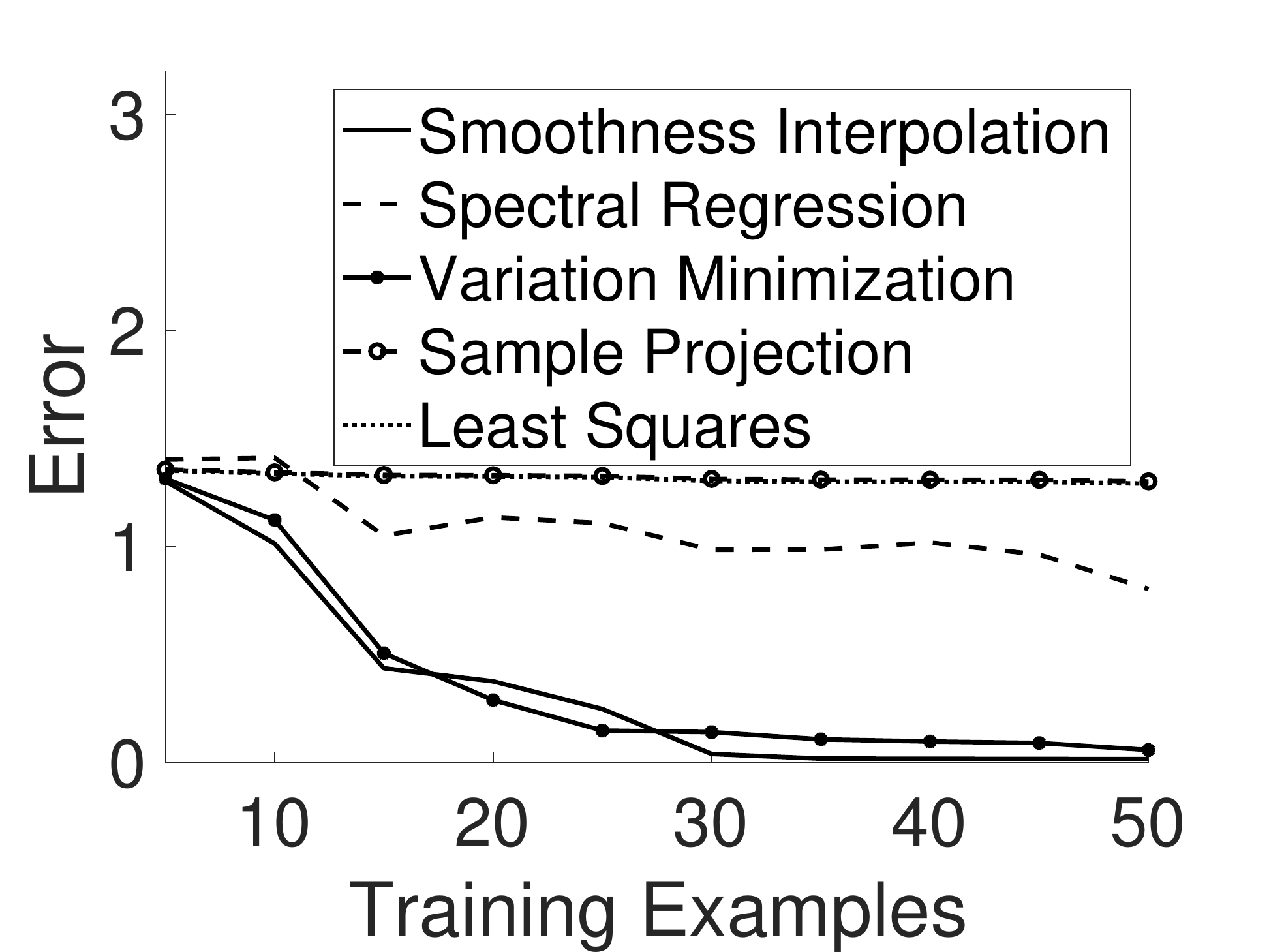}}
  \caption{Error rate~\eqref{equ:error1} for reconstruction of an approximately $20$-bandlimited graph signal. The bandlimit is assumed to be unknown. The result of~\eqref{equ:sparseOpt} is presented in blue. The results of spectral regression,~\cite{chen2015var},~\cite{chen2016reconstruction} and ~\cite{ma2015jmlr}  are presented in red, yellow, purple and green, respectively. (a) Results for sampling set selected uniformly at random. (b) Results for sampling set selected as detailed in~\cite{chen2015sampling}.}
  \label{fig:approx_bandlimit_bunny_uk}
\end{figure}

We conclude that, whether the bandlimit is known or not, our method is among the state-of-the-art methods that achieve the lowest error rate.

\subsection{MNIST dataset}
In the previous Section we presented results on synthetic data. We now compare our methods~\eqref{equ:sparseOpt} and Algorithm~\ref{alg:it} to~\cite{keller2011regression, chen2016reconstruction, chen2015var, ma2015jmlr} on the MNIST dataset~\cite{lecun1998mnist} of handwritten digits . This dataset includes $60000$ training images and $10000$ test images. We solve a clustering problem where 
the goal is to associate each image with  the  digit it depicts.
 
We formulate the clustering problem  as a graph signal interpolation problem, in the manner detailed in \cite{chen2015sampling}. 
Specifically, we select $1000$ images of each digit and represent this reduced set as an undirected weighted graph, wherein  
each image $x_i$ is represented by a single node $v_i$. 
The Euclidean distances between vectorizations 
of the  images are used as a distance measure between their respective nodes. 
We denote the matrix of pairwise distances as $\mathbf{F}$.  

We keep only the $L=12$ smallest  entries 
for each row of $\mathbf{F}$, and denote by $\mathcal{M}_n$ the indices of the $L$ smallest entries for each row $n$. 
The weight of an edge between  two images is defined as
\begin{equation}
W_{i,j} = \begin{cases}
\frac{F_{i,j} \cdot N^2}{\sum_{n=1}^{N} \sum_{m \in \mathcal{M}_n} F_{n,m}} & \text{for } j \in \mathcal{M}_i \\
0 &  \text{for } j \notin \mathcal{M}_i
\end{cases},
\label{equ:WF}
\end{equation}
where $N$ is the number of nodes in the graph. The graph shift operator is the Markov transition matrix.

We 
define ten smooth graph signals $\mathbf{s^0}, \dots, \mathbf{s^9}$ as follows,
\begin{equation}
s_i^k = \begin{cases}
1 & \text{if } {v}_i \text{ represents an image of the digit } k,\\
0 & \text{otherwise}.
\end{cases}
\end{equation}
Each of these signals is known over a subset of $r$ nodes in the training set. 
The nodes in this set are determined according to the sampling suggested by Chen \textit{et al.} (\cite{chen2015sampling} Algorithm 1).  
Our goal is to recover the  signal over the remaining $10000 - r$ nodes. In order to do this, we 
interpolate each signal $\mathbf{s^0},\dots,\mathbf{s^9}$ independently.  We then map node 
${v}_i$ to the scalar value $k \in \{0, 1, ..., 9\}$ when our interpolated signals satisfy 
$\vert s^k_i \vert > \vert s^m_i \vert$ for all $k \neq m$.  

We present in Fig. \ref{fig:jelena_comp1}   a comparison between our suggested smoothness interpolation~\eqref{equ:sparseOpt}, its iterative extension (Algorithm~\ref{alg:it}),  the interpolation suggested by Chen \textit{et al.}~\eqref{equ:chen} \cite{chen2015sampling} and the method of Jung \textit{et al.} \cite{jung2019message}.\footnote{
{We note that 
in order for this method to be successful, the boundary of clusters (or nodes close to the boundary) needs to be sampled. 
In our experimental setup, between 10 and 100 nodes are sampled, and the sampling approach does not take into account connection to the boundary. Since there are 10 clusters, 
 we do not get a good representation of the cluster borders. Therefore, the results of this experiment will improve when using more samples and a different 
 sampling strategy. In addition, as this method can be implemented in a distributive manner, it is very fast.}} 
This figure plots the percent of correctly interpolated entries of the graph signal as a function of 
the number of training examples ($r$). Mathematically, this is
\begin{equation*}
\frac{100}{N} \cdot TP,
\end{equation*}
where $TP$ are the number of correctly interpolated entries of the the graph signal. 

Clearly, our iterative algorithm achieves the highest accuracy. 
In addition, when interpolating a graph signal ${s} \in \mathbb{R}^{10000}$ from $20-30$ samples, 
this method far outperforms \cite{chen2015sampling, jung2019message}. Our non-iterative smoothness 
interpolation~\eqref{equ:sparseOpt} also outperforms \cite{chen2015sampling, jung2019message}.

\begin{figure}
  \centering
  {\includegraphics[width=0.7 \linewidth]{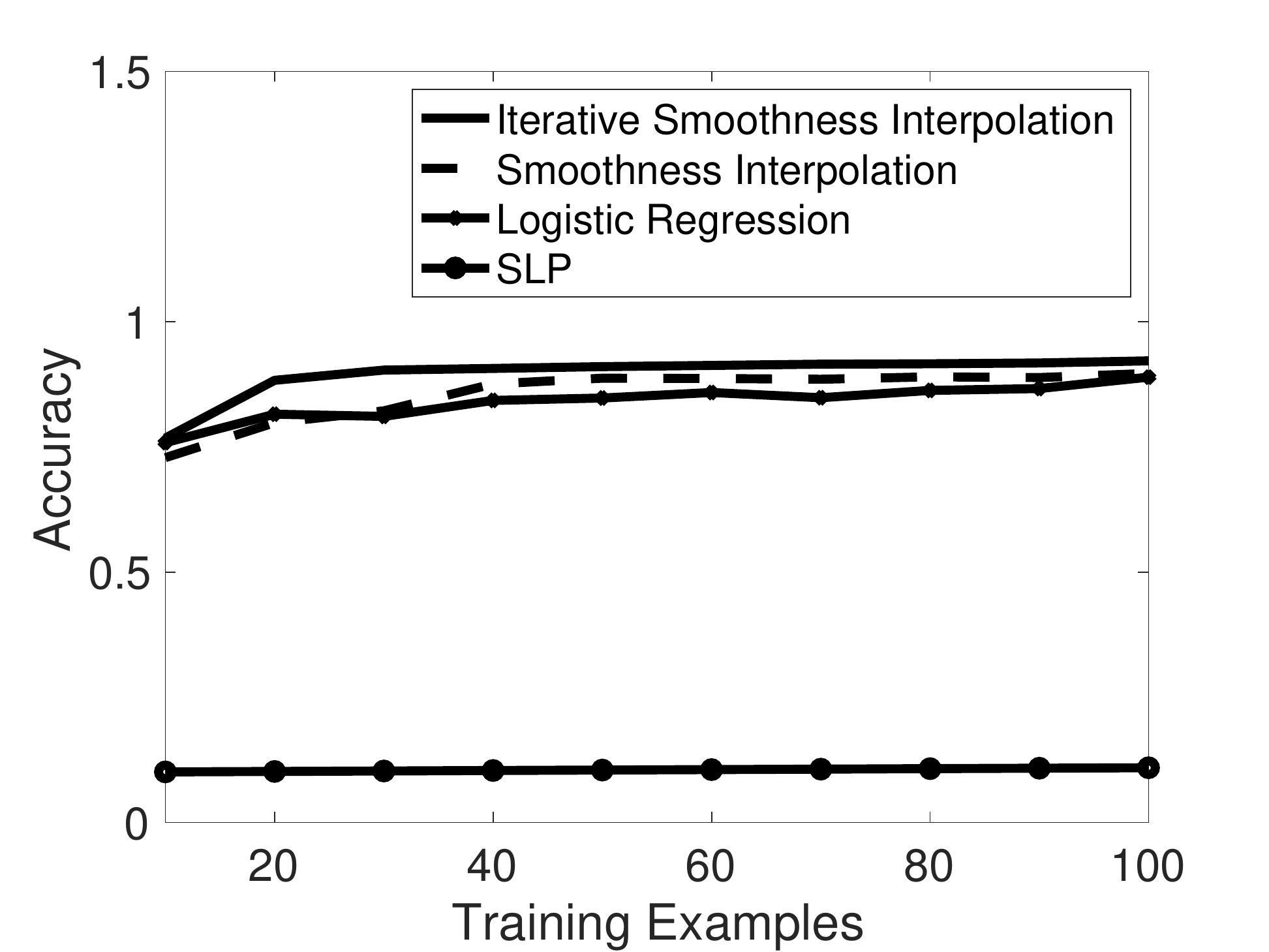} }
  \caption{Graph signal interpolation on the MNIST data set. We present size of training set ($r$) vs. accuracy of graph signal reconstruction. 
  The training set is selected according to \cite{chen2015sampling}. 
  Results for 
  our Markov variation based method~\eqref{equ:sparseOpt} are in blue. Results of the logistic regression optimization of \cite{chen2015sampling} \eqref{equ:chen}, are presented in red. Results for the total variation-based optimization \cite{chen2015sampling} are presented in yellow. Results for our iterative 
  Markov variation based method (Algorithm \ref{alg:it}) are presented in purple.}
  \label{fig:jelena_comp1}
\end{figure}

Next, we show the speed-up of our Nystr\"{o}m-based interpolation method in comparison with~\eqref{equ:sparseOpt} and~\eqref{equ:sparseOpt2}. 
For this comparison we 
build a graph from the full $70000$ training and test images in the MNIST dataset. The graph is built as described above, with two 
small modifications. First, we keep $L=200$ nearest neighbors for each node. In addition, 
we  symmetrize the affinity matrix~\eqref{equ:WF} as
\begin{equation}
W^s_{i,j} = \max \left(W_{i,j}, W_{j, i} \right)
\label{equ:symW}
\end{equation}
before calculating the graph shift operator $\mathbf{P}=\mathbf{D}^{-1} \mathbf{W}^s$.

Fig. \ref{fig:results1} presents a comparison of  accuracy and runtime between 
our Nystr\"{o}m-based interpolation method,~\eqref{equ:sparseOpt} and~\eqref{equ:sparseOpt2}.  
As in Fig. \ref{fig:jelena_comp1}, the percent of correctly interpolated entries of the graph signal is presented as a function of 
the number of training examples ($r$).  
Since two of these methods 
calculate the eigenvectors exactly, and since we cannot compute the eigendecomposition of a $70000 \times 70000$ Markov 
 matrix, we reduce the size of the graph to  $r$  randomly chosen (sampled)  nodes from the training set and all $10000$ nodes from the test set.  
 The interpolation is done over this reduced graph, and results in an estimation of the graph signal over the test set. 

\begin{figure}
  \centering
  \subfigure[]{\includegraphics[width=0.7 \linewidth]{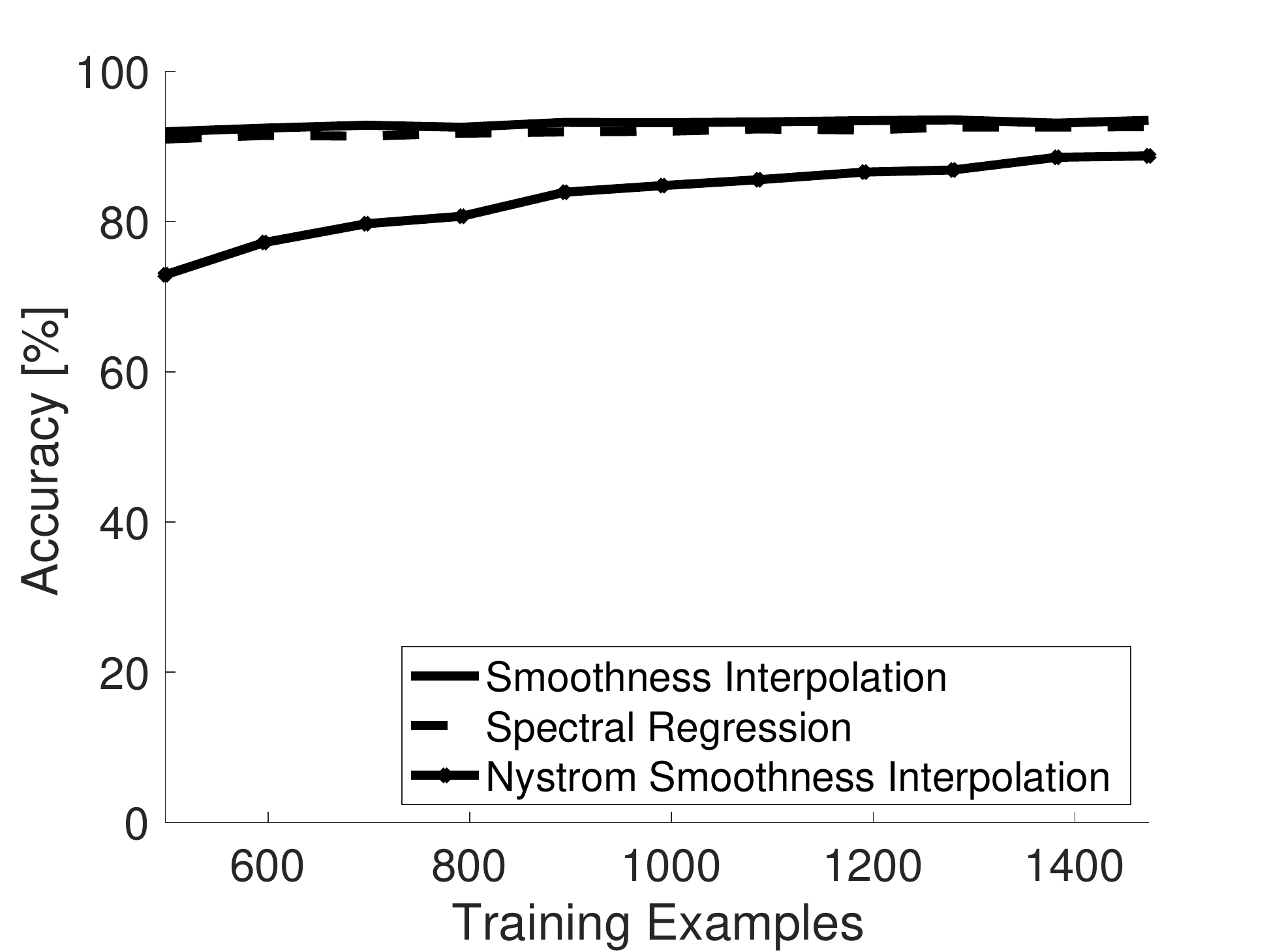} }
   \subfigure[]{\includegraphics[width=0.7 \linewidth]{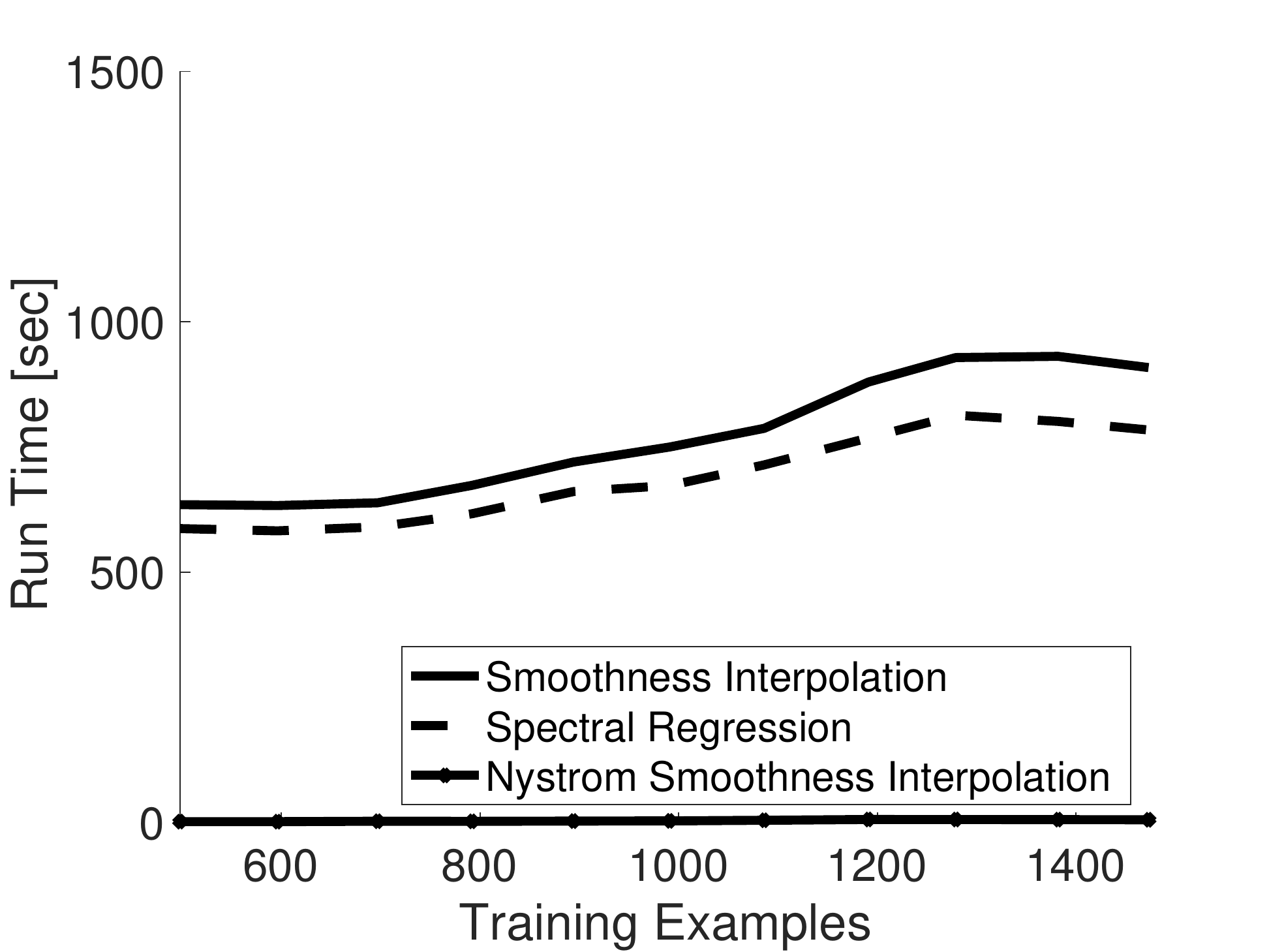} }
  \caption{Graph signal interpolation on the MNIST data set. Results for the Nystr\"{o}m optimization method are presented in blue. Results of 
  our Markov variation based method~\eqref{equ:sparseOpt} are in red. Results of spectral regression~\eqref{equ:sparseOpt2} are presented in yellow. (a) size of training set ($r$) vs. accuracy graph signal reconstruction.  
  (b) size of training set ($r$) vs. total time for graph signal reconstruction. }
  \label{fig:results1}
\end{figure}

Our interpolation method~\eqref{equ:sparseOpt} consistently achieves the highest accuracy, while our Nystr\"{o}m-based 
method has  reduced run time while maintaining high accuracy.

 In summary, our smoothness interpolation method~\eqref{equ:sparseOpt}, and its iterative extension (Algorithm \ref{alg:it}) outperform 
 the graph signal interpolation methods of  \cite{chen2015sampling, jung2019message, keller2011regression} on the MNIST dataset. In addition, our 
 variation on the  Nystr\"{o}m extension achieves good accuracy while allowing to quickly interpolate many entries of very large 
 graph signals. 
 
 We note that in the above experiments the optimization problem~\eqref{equ:sparseOpt} was solved over the $r$ leading eigenvectors. This 
 is due to the fact that the Nystr\"{o}m smoothness interpolation is limited to $r$ eigenvectors.

\subsection{Temperature measurements}

In the previous section we dealt with a clustering problem, where the graph signal was a labeling of the nodes. In general 
similar clusters need not 
have similar labels. For example, while the digits 3 and 8 are similar, their labels are not. We now consider a regression problem, where the graph signal 
is a quantity rather than a label. 

We turn to a dataset of average temperatures as measured by $2181$ sensors across the contiguous United States on January 1st, 2011 \cite{gsod2011dataset}. 
The dataset contains both longitude, latitude and elevation of each sensor. Following \cite{Sandryhaila2013shift}, we represent each sensor as a node in a K-nearest 
neighbors graph. Edge weights are defined according to (\cite{Sandryhaila2013shift} eq (26)), 
\begin{equation}
A_{n,m} = \frac{e^{-d_{n,m}^2}}{\sqrt{\sum_{k \in \mathcal{N}_n} e^{-d^2_{n,k}} \sum_{l \in \mathcal{N}_m} e^{-d^2_{n,l}}}},
\end{equation}
where $d_{n,m}$ denotes the geodesic distance between node $n$ and node $m$. As we restrict the discussion to undirected graphs, 
the affinity matrix is, 
\begin{equation}
W_{n,m} =\max \left(A_{n,m}, A_{m,n} \right).
\end{equation}

Fig. \ref{fig:results3}  presents a comparison in terms of error between our iterative optimization method  (Algorithm~\ref{alg:it})  and 
spectral regression \cite{keller2011regression} as well as \cite{ma2015jmlr, chen2016reconstruction, chen2015var}. As this is a relatively small dataset, there is no need to compare to our fast interpolation method. Additionally, in each method, we do not make any assumption on the bandlimit of the graph signal. 
We compute the error as 
\begin{equation}
\Vert \mathbf{y} - \mathbf{\hat{y}} \Vert_2 / \Vert \mathbf{y} \Vert_2
\end{equation}
where $\mathbf{y}$ is the true graph signal and $\mathbf{\hat{y}}$ is the interpolation. 

In Fig.~\ref{fig:results4}  we provide a comparison when we assume a bandlimit of $9$ in the methods~\cite{keller2011regression, chen2016reconstruction, ma2015jmlr}. It is clear from Figs.~\ref{fig:results3} and \ref{fig:results4} that, when the dataset is small enough that exact computation of the  eigendecomposition of the graph shift 
matrix is feasible,  our suggested interpolation method far 
outperforms  state-of-the-art methods~\cite{keller2011regression, ma2015jmlr, chen2016reconstruction, chen2015var}. 

\begin{figure}
  \centering
{\includegraphics[width=0.7 \linewidth]{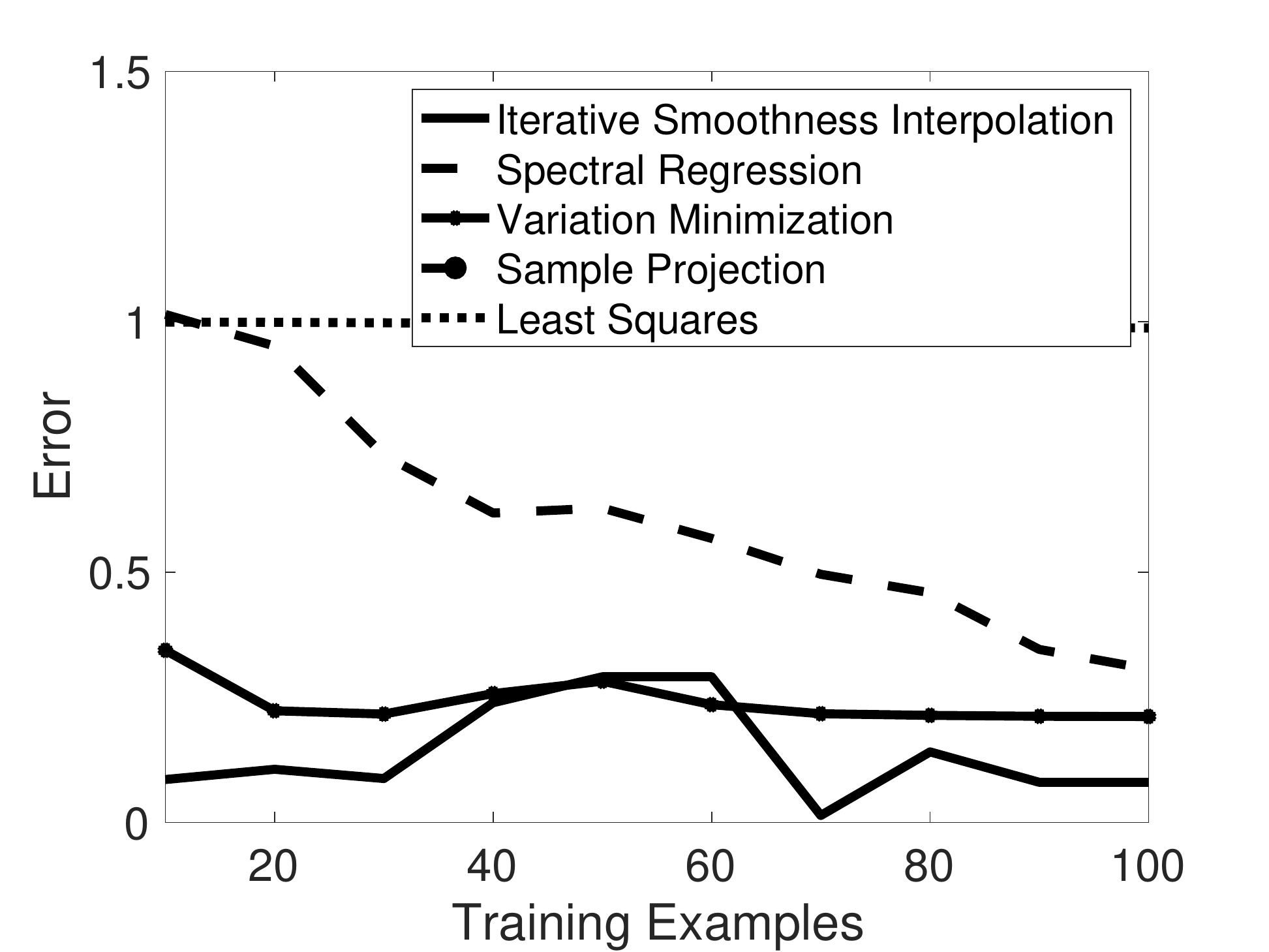}}
  \caption{Graph signal interpolation over $2151$ weather sensors scattered across the contiguous United States. 
  The graph was built using $K=10$ nearest neighbors. No assumptions were made about the bandlimit of the signal. The result of Algorithm~\ref{alg:it} is presented in blue. The results of spectral regression,~\cite{chen2015var},~\cite{chen2016reconstruction} and ~\cite{ma2015jmlr}  are presented in red, yellow, purple and green, respectively.}
  \label{fig:results3}
\end{figure}

\begin{figure}
  \centering
{\includegraphics[width=0.7 \linewidth]{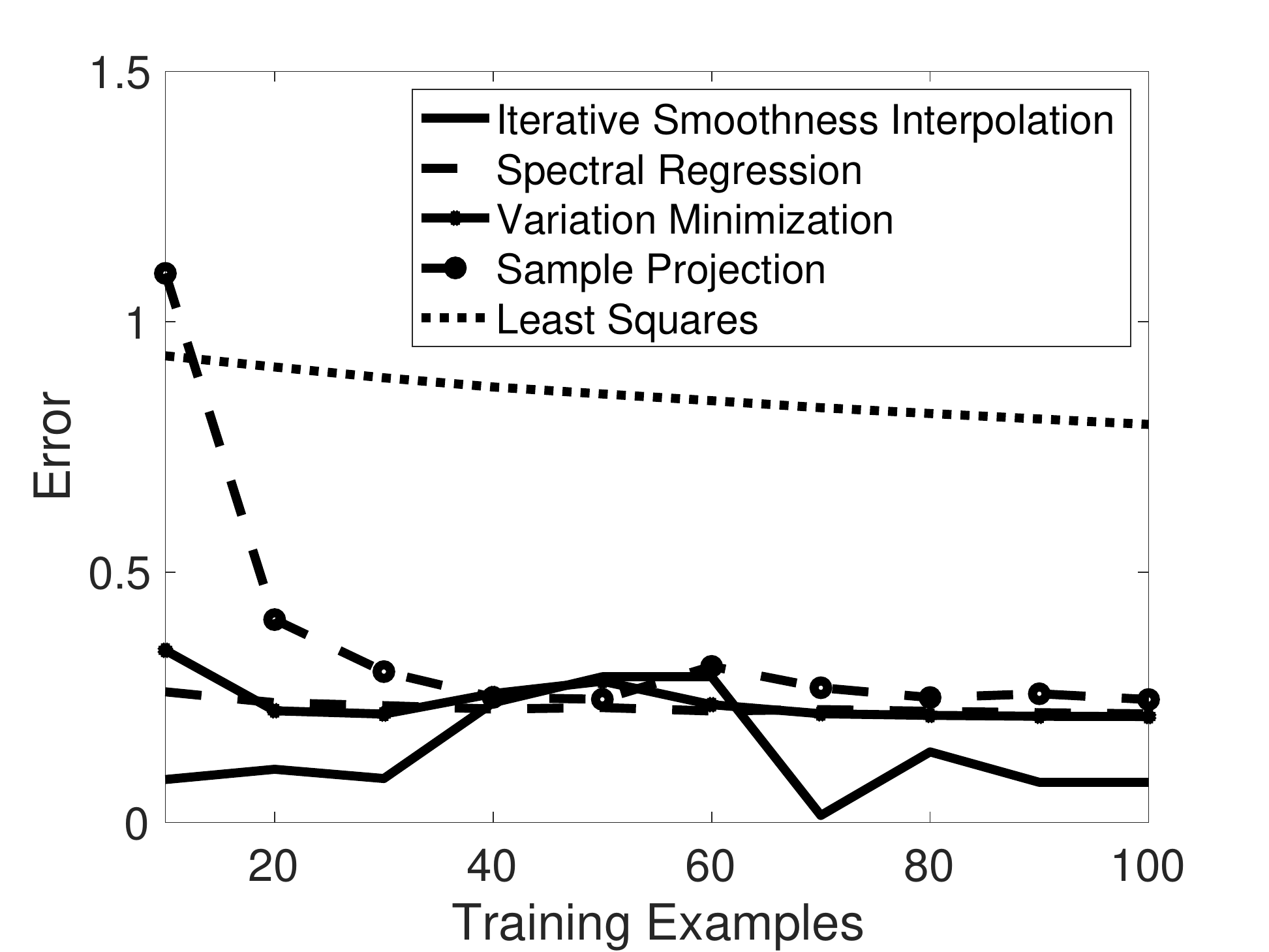}}
  \caption{Graph signal interpolation over $2151$ weather sensors scattered across the contiguous United States. 
  The graph was built using $K=10$ nearest neighbors. The graph signal is assumed to have a bandlimit of $9$. The result of Algorithm~\ref{alg:it} is presented in blue. The results of spectral regression,~\cite{chen2015var},~\cite{chen2016reconstruction} and ~\cite{ma2015jmlr}  are presented in red, yellow, purple and green, respectively.}
  \label{fig:results4}
\end{figure}

\section{Conclusion}

The field of signal processing on graphs 
strives to generalize definitions and operations from
signal processing to data represented by a graph. An important definition in this 
field is the graph shift operator.  In 
this paper we define the graph shift operator to be the
Markov matrix and use this definition to formulate the Markov variation,  a   
smoothness measure for graph signals. This measure is closely related to the diffusion embedding vectors of the 
nodes of the graph. 

We use the Markov variation to derive a method for interpolation of smooth graph signals. This is done by defining a system of linear 
equations derived from the Markov variation. Since this system may be underdetermined, we select the solution with minimal $l_1$ norm.  
{This method naturally extends to an iterative interpolation, where each iteration either leaves the solution unchanged or returns a smoother 
solution. }
We experimentally verify our interpolation methods over the MNIST dataset of handwritten digits and over a dataset of temperature 
measurements across the contiguous United States. We show {that} our method outperforms state-of-the-art interpolation methods such 
as \cite{keller2011regression} on both these data sets. 

In addition, we utilize the Nystr\"{o}m extension for a computationally efficient solution of the aforementioned minimization problem.  
We show that our efficient approximation achieves good results on the MNIST data set in {greatly} reduced runtimes.

\section*{Ackowledgements}
The authors would like to thank Antonio Ortega for his helpful discussions on the nature of graph signals.

\section*{Appendix A}

In this Appendix we provide a summary of some key properties of the Markov matrix along with their proof. These properties were used in the analysis of the Markov variation and the derivation of the graph signal interpolation method~\eqref{equ:sparseOpt}.

\begin{lemma} 
The Markov matrix of a connected and undirected graph is diagonalizable.
\end{lemma}
\begin{proof}
The normalized graph Laplacian is defined as
\begin{equation}
\label{equ:define}
\mathbf{L} = \mathbf{D}^{-\frac{1}{2}} \left( \mathbf{D} - \mathbf{W} \right) \mathbf{D}^{-\frac{1}{2}} = \mathbf{I}_N - \mathbf{D}^{-\frac{1}{2}} \mathbf{W} \mathbf{D}^{-\frac{1}{2}},
\end{equation}
and the Markov matrix is defined as 
\begin{equation} \label{equ:similar}
\mathbf{P} = \mathbf{D}^{-1} \mathbf{W} =  \mathbf{D}^{-\frac{1}{2}} \left(  \mathbf{D}^{-\frac{1}{2}} \mathbf{W}  \mathbf{D}^{-\frac{1}{2}} \right)  \mathbf{D}^{\frac{1}{2}}.
\end{equation}
It follows from (\ref{equ:define}) and (\ref{equ:similar}) that
\begin{equation}
\label{equ:similar3}
 \mathbf{P} =  \mathbf{D}^{-\frac{1}{2}} \left(  \mathbf{I}_N -  \mathbf{L} \right)  \mathbf{D}^{\frac{1}{2}}.
\end{equation}
This means that the Markov matrix is similar to $\mathbf{I}_N - \mathbf{L}$. The normalized graph Laplacian is a symmetric matrix and is thus diagonalizable. 
The same is true for $\mathbf{I}_N - \mathbf{L}$. As $\mathbf{P}$ is similar to a diagonalizable matrix, it is also diagonalizable.
\end{proof}

Since $\mathbf{P} \in \mathbb{R}^{N \times N}$ is diagonalizable, it has $N$ eigenvalues and $N$ eigenvectors. In the proposition below we 
denote the eigenvectors and eigenvalues of $\mathbf{P}$ as  $\{ \psi_i \}_{i=1}^{N}$  and 
 $\{ \lambda_i \}_{i=1}^{N}$, respectively. We further denote the eigenvectors and eigenvalues of $\mathbf{L}$ 
as $\{ \mathbf{u}_i \}_{i=1}^{N}$ and   $\{ \tilde{\lambda_i} \}_{i=1}^{N}$. 

\begin{proposition} The eigenvectors and eigenvalues of the Markov matrix obey the following: 

\begin{enumerate}
\item {\label{prop:1}
$\psi_i = \mathbf{D}^{-\frac{1}{2}} \mathbf{u}_i, \quad \lambda_i = 1 - \tilde{\lambda}_i.$}

\item {\label{prop:2}
$\vert \lambda_i \vert \le 1.$}

\item {\label{prop:5}
The leading eigenvector of the Markov matrix is constant.}

\end{enumerate}
\label{prop:connection}
\end{proposition}
\begin{proof}
\begin{enumerate}
\item 
In the proof of Lemma 1 we saw that $\mathbf{P}$ is similar to $\mathbf{I}_N - \mathbf{L}$. It follows that $\mathbf{L}$ is similar to $\mathbf{I} - \mathbf{P}$,
\begin{equation}
\mathbf{I}_N -\mathbf{P} =\mathbf{D}^{-\frac{1}{2}} \mathbf{L}  \mathbf{D}^{\frac{1}{2}},
\end{equation}
and that
\begin{equation}
 \left( \mathbf{I}_N -\mathbf{P} \right) \mathbf{D}^{-\frac{1}{2}} = \mathbf{D}^{-\frac{1}{2}} \mathbf{L}.
 \end{equation}

It follows that
\begin{equation}
\left( \mathbf{I}_N -\mathbf{P} \right) \mathbf{D}^{-\frac{1}{2}}  \mathbf{u}_i= \tilde{\lambda}_i \mathbf{D}^{-\frac{1}{2}} \mathbf{u}_i.
\end{equation}
Thus, $\mathbf{D}^{-\frac{1}{2}} \mathbf{u}_i$ is an eigenvector of $\mathbf{I}_N -\mathbf{P}$, with eigenvalue $\tilde{\lambda}_i$. 
Also, 
\begin{equation}
\left( \mathbf{I}_N -\mathbf{P} \right) \mathbf{D}^{-\frac{1}{2}}  \mathbf{u}_i=  \mathbf{D}^{-\frac{1}{2}}  \mathbf{u}  -\mathbf{P} \mathbf{D}^{-\frac{1}{2}}  \mathbf{u}_i =
\tilde{\lambda}_i \mathbf{D}^{-\frac{1}{2}} \mathbf{u}_i,
\end{equation}
so that
\begin{equation} \label{equ:UtoV}
\mathbf{P} \mathbf{D}^{-\frac{1}{2}} \mathbf{u}_i =
\left( 1 - \tilde{\lambda}_i \right) \mathbf{D}^{-\frac{1}{2}} \mathbf{u}_i.
\end{equation}
This proves that if $\mathbf{u}_i$ is an eigenvector of $\mathbf{L}$ with eigenvalue $\tilde{\lambda}_i$, then 
$\mathbf{D}^{-\frac{1}{2}} \mathbf{u}_i$ is an eigenvector of $\mathbf{P}$ with eigenvalue $1-\tilde{\lambda}_i$.

\item 

We know from Proposition 1 part 1  that
\begin{equation}\label{equ:lambda2}
  \lambda_i = 1 - \tilde{\lambda}_i.
\end{equation}
Let $\tilde{\lambda}_1$ be the smallest valued eigenvalue of the normalized graph Laplacian. As the normalized graph Laplacian is a positive semi definite matrix it follows that $\tilde{\lambda}_1 \ge 0$ and thus
$\lambda_1 \le 1$, where $\lambda_1$ is the largest eigenvalue of the Markov matrix.

Let $\tilde{\lambda}_N$ be the largest valued eigenvalue of the normalized graph Laplacian.
Chung \cite{chung1997book} used the rayleigh quotient to prove that $\tilde{\lambda}_N \le 2$.  
Therefore, $\lambda_N \ge -1$, where $\lambda_N$ is the smallest eigenvalues of the Markov matrix.

\item

To prove part 3, we note that each row in the Markov matrix sums to $1$. Thus,
\begin{equation}
\mathbf{P} \mathbf{1} = 1 \cdot \mathbf{1},
\end{equation}
where $\mathbf{1}$ is the all ones vector. We see that $1$ is an eigenvalue of $\mathbf{P}$, and is associated with a constant eigenvector. 
We know from part $2$ of the proposition that the eigenvalues are upper bounded by $1$. Therefore, the constant 
eigenvector must be the leading one.
\end{enumerate}
\vspace{-0.38cm}
\end{proof}

\section*{Appendix B}
\label{appendix:local_smoothness}
\begin{lemma} 
 Any 
solution of the system of equations~\eqref{equ:underdetermine} is guaranteed to be smooth over the one-hop neighborhoods of the nodes in $\mathcal{M}$.
\end{lemma}
\begin{proof}
The right-hand side of~\eqref{equ:underdetermine} can be written as
\begin{equation}
\begin{bmatrix} 
\lambda_1 \psi_1 \left( \mathcal{M} \right) & \cdots & \lambda_N \psi_N \left( \mathcal{M}  \right)\\
\end{bmatrix} 
\hat{\mathbf{s}} = \mathbf{V} \left( \mathcal{M} \right) \mathbf{\Lambda} \hat{\mathbf{s}},
\end{equation}
where $ \mathbf{V}$ is the matrix of eigenvectors of the graph's Markov matrix $\mathbf{P}$, $\mathbf{V} \left( \mathcal{M} \right) $ 
are the rows of $\mathbf{V}$ corresponding to the indices of the sampled nodes, and $\mathbf{\Lambda}$ is the diagonal matrix of eigenvalues 
of the Markov matrix. Using the graph Fourier transform~\eqref{equ:gft}, we get
\begin{equation}
\mathbf{V} \left( \mathcal{M} \right) \mathbf{\Lambda}  \hat{\mathbf{s}} = \mathbf{V} \left( \mathcal{M} \right) \mathbf{\Lambda} \mathbf{V}^{-1} \mathbf{s}.
\end{equation}
As $\mathbf{P} = \mathbf{V}\mathbf{\Lambda} \mathbf{V}^{-1}$, clearly, 
\begin{equation}
 \mathbf{V} \left( \mathcal{M} \right) \mathbf{\Lambda} \mathbf{V}^{-1} \mathbf{s} = \mathbf{P} \left( \mathcal{M} \right) \mathbf{s}.
\end{equation}
Therefore,~\eqref{equ:underdetermine} is equivalent to the following system
\begin{equation}
\mathbf{s}_{\mathcal{M}} = \mathbf{P} \left( \mathcal{M} \right) \mathbf{s},
\end{equation}
and is satisfied only if the graph signal is smooth in the neighborhood of the sampled nodes.
\end{proof}

\section*{Appendix C}
In Section \ref{subsec:bigData} we modified the Nystr\"{o}m extension from 
\begin{equation} 
\label{equ:appendix1}
\tilde{\mathbf{Z}} = \begin{bmatrix} \mathbf{Z} \\
\mathbf{B} \mathbf{Z} \mathbf{Q}^{-1}
\end{bmatrix}
\end{equation}
to 
\begin{equation} 
\label{equ:appendix2}
\tilde{\mathbf{Z}} = \begin{bmatrix} \mathbf{Z} \\
\mathbf{B} \mathbf{Z}
\end{bmatrix}.
\end{equation}
In order to justify this, we examine the approximation of the eigenvectors of $\mathbf{P}$ which can be found using~\eqref{equ:appendix1},
\begin{equation} 
\label{equ:appendix3}
\mathbf{D}^{-\frac{1}{2}} \tilde{\mathbf{Z}} = \mathbf{D}^{-\frac{1}{2}} \begin{bmatrix} \mathbf{Z} \\
\mathbf{B} \mathbf{Z} \mathbf{Q}^{-1}
\end{bmatrix}.
\end{equation}
We decompose the diagonal matrix $\mathbf{D}$ as 
\begin{equation*} 
{\mathbf{D}} = \begin{bmatrix} \mathbf{D}_e & \mathbf{0} \\ \mathbf{0} & 
\mathbf{D}_b 
\end{bmatrix},
\end{equation*}
where $\mathbf{D}_e \in \mathbb{R}^{r \times r}$.  Substituting  into~\eqref{equ:appendix3}, the approximation 
of the eigenvectors of $\mathbf{P}$ can be expressed as
\begin{equation} 
\label{equ:appendix4}
\mathbf{D}^{-\frac{1}{2}} \tilde{\mathbf{Z}} =  \begin{bmatrix} \mathbf{D}_e^{-\frac{1}{2}} \mathbf{Z} \\
 \mathbf{D}_b^{-\frac{1}{2}} \mathbf{B} \mathbf{Z} \mathbf{Q}^{-1}
\end{bmatrix}.
\end{equation}
We further denote $\mathbf{A} = \mathbf{D}_b^{-\frac{1}{2}} \mathbf{B} \mathbf{Z}$ and examine $\mathbf{A} \mathbf{Q}^{-1}$. 
As $\mathbf{Q}$ is a diagonal matrix, 
\begin{equation} 
\label{equ:appendix5}
\mathbf{A} \mathbf{Q}^{-1} = \begin{bmatrix}
\frac{\mathbf{a}_1}{Q_{1,1}} & \frac{\mathbf{a}_2}{Q_{2,2}} & \cdots \frac{\mathbf{a}_r}{Q_{r,r}}
\end{bmatrix},
\end{equation}
where $\mathbf{a}_i$ denotes the $i$th column of $\mathbf{A}$. That is, rows $r+1$ through $N$ of  the $i$th eigenvector of $\mathbf{P}$ are multiplied by the inverse 
of the $i$th eigenvalue of $\mathbf{E}$~\eqref{equ:structure}. 

{We assume} that $\mathcal{M} = \{1, \dots, r \}$. That is, since the numbering of nodes is arbitrary, when creating the graph shift we 
 assign the first $r$ rows and $r$ columns to  the sampled nodes. In this case, 
the optimization problem we solve is
\begin{multline} \label{equ:appendix6}
{\mathbf{x}} = \arg \underset{\mathbf{y}}{\min} \Vert \mathbf{y} \Vert_0 \quad \text{such that} \\
\begin{bmatrix} 
\left( 1-Q_{1,1}\right) {\mathbf{z}}_1 &  \cdots &\left( 1-Q_{r,r}\right) {\mathbf{z}}_r 
\end{bmatrix} 
\mathbf{y} = \mathbf{s}_{\mathcal{M}},
\end{multline}
where ${\mathbf{z}}_i$ is the $i$th column of $\mathbf{D}_e^{-\frac{1}{2}} \mathbf{Z}$. We note that 
the optimization problem~\eqref{equ:appendix6} depends only upon the first $r$ rows 
of $\mathbf{D}^{-\frac{1}{2}} \tilde{\mathbf{Z}}$. As those rows are unaffected by the approximation~\eqref{equ:appendix2}, it is clear that 
the approximation does not affect the solution $\mathbf{y}$.

In addition, we have found that for smooth graph signals, the entries in the spectrum of the signal 
that correspond to the lower valued eigenvalues of $\mathbf{P}$ are negligible. This means that for most of the eigenvectors of the Laplacian, 
it makes no difference how we approximate their eigenvalues since they will be ignored in the interpolation process. The eigenvectors that 
are not ignored correspond to the higher valued eigenvalues of $\mathbf{P}$. So, in essence, the approximation of~\eqref{equ:appendix2} 
just means that we assume that the higher valued eigenvalues of $\mathbf{P}$ are approximately equal. While this assumption is not strictly 
correct, it does prevent the low eigenvalues of the laplacian (which correspond to the high eigenvalues of $\mathbf{P}$) from causing numerical 
instabilities.

\bibliographystyle{ieeetran}
\bibliography{SPG}
}

\end{document}